\documentclass[12pt, draftclsnofoot, onecolumn]{IEEEtran}
\usepackage{tipa}
\usepackage{times}
\usepackage{amsmath}
\usepackage{amssymb}
\usepackage{graphicx}
\usepackage{cases}
\usepackage{subfigure}
\usepackage{color}
\usepackage{soul}
\usepackage{multirow}
\usepackage{algorithm}
\usepackage{algorithmic}
\usepackage{thmbox}
\usepackage{bm}
\usepackage{cite}
\usepackage{hyperref}
\ifCLASSINFOpdf
\else

\fi
\newtheorem{theorem}{Theorem}
\newtheorem{remark}{Remark}

\newtheorem{Proposition}{Proposition}

\newcommand{\tabincell}[2]{\begin{tabular}{@{}#1@{}}#2\end{tabular}}

\begin{document}

\title{Robust Localization Using Range Measurements with Unknown and Bounded Errors}

\author{
Xiufang~Shi,~\IEEEmembership{Student Member,~IEEE},
Guoqiang Mao,~\IEEEmembership{Senior Member,~IEEE},\\
Brian.D.O. Anderson,~\IEEEmembership{Life Fellow,~IEEE},
Zaiyue~Yang,~\IEEEmembership{Member,~IEEE},\\
and Jiming Chen,~\IEEEmembership{Senior Member,~IEEE}

\thanks{ X. Shi, Z. Yang and J. Chen are with the State Key Laboratory of Industrial Control Technology, Zhejiang University, Hangzhou, China. (e-mail: xfshi.zju@gmail.com, yangzy@zju.edu.cn, jmchen@ieee.org) }
\thanks{ G. Mao is with School of Computing and Communication, University of Technology Sydney, Sydney, NSW 2007, Australia, and Data61-CSIRO, Sydney, NSW 2015, Australia. (e-mail: Guoqiang.Mao@uts.edu.au)}
\thanks{ B.D.O. Anderson is Research School of Engineering, Australian National University, Canberra, ACT 0200, Australia, and Data61-CSIRO, Canberra, ACT 2601, Australia. His work was supported in part by National ICT Australia Ltd., and in part by the Australian Research Council under Grant DP-110100538
and Grant DP-130103610. (e-mail: brian.anderson@anu.edu.au)}
}

\maketitle

\begin{abstract}
Cooperative geolocation has attracted significant research interests in recent years. A large number of localization algorithms rely on the availability of statistical knowledge of measurement errors, which is often difficult to obtain in practice. Compared with the statistical knowledge of measurement errors, it can often be easier to obtain the measurement error bound. This work investigates a localization problem assuming unknown measurement error distribution except for a bound on the error. We first formulate this localization problem as an optimization problem to minimize the worst-case estimation error, which is shown to be a non-convex optimization problem. Then, relaxation is applied to transform it into a convex one. Furthermore, we propose a distributed algorithm to solve the problem, which will converge in a few iterations. Simulation results show that the proposed algorithms are more robust to large measurement errors than existing algorithms in the literature. Geometrical analysis providing additional insights is also provided.
\end{abstract}

\begin{IEEEkeywords}
Cooperative localization, bounded measurement error, worst-case estimation error, Chebyshev center, convex relaxation, semidefinite programming
\end{IEEEkeywords}

\IEEEpeerreviewmaketitle

\section{Introduction}

\IEEEPARstart{W}{ireless} sensor networks (WSNs) play an important role in many applications \cite{cao2008ISAtrans,mao2013road,TWC_zhang2013energy,TIE2014tracking,mao2009graph,ge2015energy}, such as environmental monitoring, target tracking, energy harvesting, etc. Most of these applications are location-dependent, which require knowledge of the measurement locations, and consequently make localization one of the most important technologies in WSNs. In the last decade, wireless localization technologies have undergone significant developments. Existing localization techniques can be divided into many categories depending on the measurement techniques and the localization algorithms being used. Range-based localization, which uses inter-node range measurements for location estimation, is one of the most widely used localization approaches. The range measurements can be estimated from received signal strength (RSS)\cite{mao2007wireless,mao2006ple}, time of arrival(TOA)\cite{TWC_TOAparkclosed}, time difference of arrival (TDOA)\cite{mao2007wireless,SDP_TDOA2009TSP}, and so on \cite{mao2007wireless}. One can also use RSS, TOA, etc. measurements to directly estimate the sensors' positions without first converting these measurements into range measurements\cite{twcrss,TWC_RSS2015distributed}.

For range-based sensor network localization, the localization problem is usually formulated as an optimization problem to determine sensors' positions, such that they are consistent with the inter-node range measurements and known anchors' positions. Various localization algorithms are developed to optimize some given objectives. One of the most widely used algorithms is maximum likelihood estimator (MLE), which maximizes the likelihood function of the unknown sensors' positions \cite{patwari2003TSP,patwari2005SPmagazine}. MLE relies on statistical knowledge of the measurement error. Another widely used algorithm is least squares estimator (LS)\cite{MoeWin2015TVT_LS,TWC_LSzhu2011distributed}, which minimizes the squared error between range measurements and range estimates. LS does not need the knowledge of measurement error distribution. When the measurement error follows a zero-mean Gaussian distribution, MLE and LS become the same. Both MLE and LS are non-Bayesian estimators, there are also some Bayesian estimators, e.g., minimum mean squared error (MMSE) estimator and maximum a posteriori (MAP) estimator, which treat sensors' position vector as a random variable with a priori distribution, but the posterior distribution of the position vector is usually very hard to describe \cite{MoeWin2009cooperative}.

Many position estimators, e.g., MLE, are based on the assumption that measurement error follows a specific distribution and its probability density function is known. In practice, the measurement error may deviate from the assumed distribution and vary according to the measurement environment, measurement technique and measurement device. And we often do not have statistical knowledge of measurement error. The performance of these algorithms becomes vulnerable to an inaccurate statistical knowledge of measurement error. For those algorithms that do not need the statistical knowledge of measurement error, e.g. LS based algorithm, they may perform well when measurement error is small; however, their performance can degrade significantly when the error becomes very large\cite{anderson2010formal}. The error bound is another and less demanding way to describe the property of measurement error. The bounded error assumption has been widely applied in many areas, e.g., set-theoretic estimation in system and control area \cite{schweppe1968recursive,eldar2008minimax,overview1996setMember}, wireless localization \cite{Eurasip2014upper,NLOS_TDOA2016TSP}, etc. Moreover, compared with the statistical distribution of measurement error, it is much easier to obtain the measurement error bound in many situations \cite{belforte1987boundEstimation}. Furthermore, there are already techniques to estimate the measurement error bound with small sets of data,\cite{belforte1987boundEstimation,NLOS_TDOA2016TSP}, e.g., using support vector algorithm to find the smallest sphere that the data live on.

Motivated by the above observations, in this paper, we design a localization algorithm that is robust against large measurement errors and does not need the statistical knowledge of measurement error; instead, only a measurement error bound is required. The main contributions of our work are summarized as follows:
\begin{enumerate}
  \item We first design a centralized robust localization algorithm to minimize the worst-case estimation error, which only uses the measurement error bound. Using geometrical analysis, we show that the algorithm has bounded localization error.
  \item A distributed and iterative localization algorithm is further developed. The convergence of the iterative estimation algorithm is proved. The communication cost and computational complexity of the proposed algorithm are analyzed.
  \item Extensive simulations show that the proposed algorithms are robust against large measurement errors, and the distributed counterpart can converge in a few iterations.
\end{enumerate}

The remainder of this paper is organized as follows. Section \ref{sec:related work} introduces the related work. Section \ref{sec:problem formualtion} gives the problem formulation. Section \ref{sec:minmax_algorithm} presents the proposed centralized localization algorithm. Section \ref{sec:geometrical interpretation} illustrates a geometrical interpretation of our problem and analyzes the error bound of our algorithm. Section \ref{sec:distributed alg} proposes a distributed algorithm. Simulation results are presented in Section \ref{sec:simulations}. Conclusions are drawn in Section \ref{sec:conclusion}.

\emph{Notations}: Throughout the paper, matrices and vectors are denoted by upper boldface and lower boldface letters, respectively; the operator $(\cdot)^T$ denotes the transposition; $\textrm{Tr}(\cdot)$ denotes the matrix trace; $\parallel\cdot\parallel$ stands for the $l_2$ norm; $\mid\mathcal{N}\mid$ denotes the cardinality of $\mathcal{N}$.

\section{Related Work}\label{sec:related work}

Range-based sensor network localization seeks to estimate the unknown sensors' positions that are consistent with the inter-node range measurements and the known anchors' positions\cite{kannan2008robust,biswas2006SDP_TOSN,MLE2014TSP,MDS_central,ESDP,SOCP}. As one of the most widely used localization approaches, range-based localization has attracted substantial research interests.

In practice, the range measurements are usually not error-free. If the statistical knowledge of the measurement error is available \textit{a priori}, MLE is statistically optimal since it maximizes the likelihood function of sensors' positions. However, the optimization problem based on MLE is often non-convex. Many existing methods, e.g., the gradient-based descent method for tackling an MLE problem, require good initialization to reach the global optimum; otherwise, they will fall into a local optimum. One alternative method that deals with the non-convexity of MLE is semidefinite relaxation (SDR), which relaxes the non-convex optimization problem into a convex optimization problem \cite{biswas2006SDP_TOSN,MLE2014TSP}. In \cite{biswas2006SDP_TOSN}, an MLE with SDR is proposed for network localization with Gaussian distributed noise. Simonetto and Leus \cite{MLE2014TSP} derived convex relaxation for MLE under different measurement error models.

If the statistical properties of the measurement error are unknown, the localization problem is usually formulated as an optimization problem minimizing some global cost function. To solve such a problem, a number of available algorithms have been proposed. Some of these algorithms are implemented in a centralized way. The semidefinite programming (SDP) based approach in \cite{biswas2006semidefinite} is one of the most widely used algorithms, and it seeks to minimize the error between the squared range measurements and the squared range estimates. Such an optimization problem is however non-convex. In order to obtain a solution, it is common to transform this problem into an SDP problem through SDR. To improve the computation efficiency of SDP in large networks, edge-based SDP (ESDP) and second-order cone programming (SOCP) relaxation approach were further proposed respectively in \cite{ESDP} and \cite{SOCP}. Another widely used approach is to use a multidimensional scaling (MDS) algorithm \cite{MDS_central}, in which, the localization problem is posed as an LS problem. Subject to use of a good initialization, a gradient-based method \cite{gradientMethods} is also a quick way for sensor position estimation. The above mentioned algorithms also have distributed versions, e.g., distributed SDP \cite{biswas2006SDP_TOSN}, distributed SOCP \cite{distributedSOCP,shi2010SGO}, distributed MDS \cite{MDS,distributedMDS}, and distributed gradient based methods \cite{biswas2006SDP_TOSN,MDS,gradientMethods}.
These algorithms are essentially minimizing the discrepancy between the range measurements and the range estimates. There are also some other distributed algorithms, e.g., \cite{DILOC,diao2014barycentric}. Khan \emph{et al}. \cite{DILOC} proposed a distributed iterative localization algorithm based on the use of barycentric coordinates, which requires all the unknown sensors lie in the convex hull of the anchors. Diao \emph{et al}. \cite{diao2014barycentric} proposed a more general algorithm, which does not require each sensor to lie inside the convex hull of its neighbors and can guarantee global convergence.
 Note though that none of these algorithms directly minimizes the position estimation error. Meanwhile, theoretical analysis establishing the mathematical relationship between the above objective functions and the position estimation error is still lacking. Moreover, although these algorithms typically perform well when the measurement error is small, when the error becomes large, their performance cannot be guaranteed.

Regarding the direct minimization of estimation error, Eldar \emph{et al}.\cite{eldar2008minimax} investigated a minimax estimator, which minimizes the worst-case estimation error, for parameter estimation in a classical linear regression model. In the problem of \cite{eldar2008minimax}, the measurement model is linear with bounded error and the true parameter vector is assumed to lie in the intersection of some known ellipsoids. Furthermore, simulations in \cite{eldar2008minimax} show the advantage of the proposed minimax estimator over the constrained least squares (CLS), which minimizes the data error. Different from the linear regression problem in \cite{eldar2008minimax}, in the localization problem, the range measurements are nonlinear functions of sensors' positions and the feasible set of the positions is not convex, which makes the localization problem more challenging. Inspired by the work in \cite{eldar2008minimax}, we design a position estimator that minimizes the worst-case position estimation error for robust localization.

\section{Problem Formulation}\label{sec:problem formualtion}

 We consider a static network in two-dimensional space, which consists of $n$ sensors, denoted by $\mathcal{V}_x=\{1,\cdots,n\}$, and $m$ anchors, denoted by $\mathcal{V}_a=\{n+1,\cdots,n+m\}$.
The true position of sensor $i$ is $\textbf{x}_i = [{x_i},{y_i}], i\in \mathcal{V}_x$, which is unknown and needs to be estimated. The position of anchor $k$ is known as ${\textbf{a}_k} = {[a_{kx},a_{ky}]}, k\in \mathcal{V}_a$.
Due to the communication limitation, a pair of nodes can acquire the range measurement between them only when they are within a certain sensing range. Let $\mathcal{N}_x$ and $\mathcal{N}_a$ respectively denote the set of sensor-sensor links and the set of sensor-anchor links, from which we can obtain range measurements. All the nodes $\mathcal{V}=\{\mathcal{V}_x,\mathcal{V}_a\}$ and the inter-node links $\mathcal{E}=\{\mathcal{N}_x, \mathcal{N}_a\}$ constitute an undirected graph $G = (\mathcal{V},\mathcal{E})$. We assume this graph is connected. Furthermore, to guarantee that all the sensors can be uniquely localized, we assume this graph is globally rigid, and there exist at least three non-collinear anchors in the area \cite{anderson2010formal,rigidity2004noise_free}. (This requirement is relevant to a two-dimensional ambient space. Four noncoplanar anchors are required for a three-dimensional ambient space.) Henceforth, we restrict attention to the two-dimensional case.

The sensor-sensor range measurement is
\begin{equation}\label{eq:z_ij}
{z_{ij}} = d_{ij} + {\upsilon _{ij}}, ~(i,j)\in \mathcal{N}_x
\end{equation}
where $d_{ij}={\parallel \textbf{x}_i- \textbf{x}_j\parallel}$ is the true distance between sensor $i$ and $j$, and $\left|{{\upsilon _{ij}}}\right| \le \gamma$ is the unknown and bounded measurement error. Correspondingly, the sensor-anchor range measurement is
\begin{equation}\label{eq:z_ik}
{z_{ik}} =  d_{ik}+ {\upsilon _{ik}}, ~(i,k)\in \mathcal{N}_a
\end{equation}
where $d_{ik}={\parallel \textbf{x}_i- \textbf{a}_k\parallel}$ is the true distance between sensor $i$ and anchor $k$, and $ \left|{{\upsilon _{ik}}}\right| \le \gamma$ is the unknown and bounded measurement error. We assume the error bound $\gamma$ is known \textit{a priori} and the measurement errors are independent of each other.

From the constraints on the range measurements, we can say that sensor $i$ lies in the following closed feasible set
\begin{align}
\mathcal{C}_i=\{\textbf{x}_i:\underline{d_{ij}}  \le \parallel {{\textbf{x}_i} - {\textbf{x}_j}} \parallel\leq \overline{d_{ij}} ,~~~\forall (i,j) \in {{\cal N}_x} \label{eq:feasible set(a)}\\
   \underline{d_{ik}} \le \parallel {{\textbf{x}_i} - {\textbf{a}_k}} \parallel \le\overline{d_{ik}} ,~~\forall (i,k) \in {{\cal N}_a}\}\label{eq:feasible set(b)}
\end{align}
where $\underline{d_{ij}}=z_{ij} - \gamma$, $\overline{d_{ij}}=z_{ij} +\gamma$, $\underline{d_{ik}}=z_{ik} - \gamma$ and $\overline{d_{ik}}=z_{ik} +\gamma$. It is possible that $\underline{d_{ij}}$ or $\underline{d_{ik}}$ becomes negative if $\gamma$ is large or very loose. In such case, we set $\underline{d_{ij}}=0$, $\underline{d_{ik}}=0$.

Since the true positions of the sensors are unknown, we cannot minimize the position estimation error directly. Therefore, we \textit{minimize the worst-case estimation error over all the feasible positions}. Since $\textbf{x}_i$ is unknown and fixed, to avoid misunderstanding below, we introduce a new variable ${\textbf{y}_i}$, which denotes an arbitrary point in the feasible set of  $\textbf{x}_i$. Let ${\hat {\textbf{x}}_i}$ denote the position estimate of sensor $i$; the worst-case estimation error would be
%\begin{equation}
%
%\end{equation}
\begin{subequations}
\begin{align}
&\mathop {\max }\limits_{{\textbf{y}_i}} \sum\limits_{i = 1}^n {{{\parallel {{\textbf{y}_i} - {\hat{\textbf{x}}_i}} \parallel}^2}}\nonumber\\
&s.t.~ ~  \underline{d_{ij}}  \le \parallel {{\textbf{y}_i} - {\textbf{y}_j}} \parallel\leq \overline{d_{ij}} ,~~~\forall (i,j) \in {{\cal N}_x}\label{subeq:dij}\\
&  ~~~~~~ \underline{d_{ik}} \le \parallel {{\textbf{y}_i} - {\textbf{a}_k}} \parallel \le\overline{d_{ik}} ,~~\forall (i,k) \in {{\cal N}_a} \label{subeq:dik}
\end{align}
\end{subequations}
We must choose ${\hat {\textbf{x}}_i}$ to minimize this worst-case error. Therefore, we seek to solve:
\begin{align}\label{eq:minmax_x}
&\mathop {\min }\limits_{{\hat {\textbf{x}}_i}} \mathop {\max }\limits_{{\textbf{y}_i}} \sum\limits_{i = 1}^n {{{\parallel {{\textbf{y}_i} - {\hat {\textbf{x}}_i}} \parallel}^2}} \nonumber\\
&s.t.~ ~ \eqref{subeq:dij}\eqref{subeq:dik}
\end{align}

To facilitate the notation and analysis, we write \eqref{eq:minmax_x} in a compact matrix form. Let $\textbf{x}=[{\textbf{x}_1},{\textbf{x}_2}, \cdots ,{\textbf{x}_n}]^T\in\mathbb{R}^{2n}$ be the true position vector, which is unknown and fixed. We also introduce a new position vector $\textbf{y}=[{\textbf{y}_1},{\textbf{y}_2}, \cdots ,{\textbf{y}_n}]^T\in\mathbb{R}^{2n}$, which denotes an arbitrary possible value of $\textbf{x}$ in the feasible set. Clearly, ${\textbf{x} \in \mathcal{C}}$ below. Let $\hat{\textbf{{x}}}=[{\hat{\textbf{{x}}}_1},{\hat{\textbf{x}}_2}, \cdots ,{\hat{\textbf{x}}_n}]^T\in\mathbb{R}^{2n}$ denote the estimate of $\textbf{x}$. Then our localization problem can be expressed as
\begin{align}\label{eq:minmax_obj}
&\mathop {\min }\limits_{\hat{\textbf{x}}} \mathop {\max }\limits_{\textbf{y} \in \mathcal{C}} \textrm{Tr}({(\textbf{y} - \hat{\textbf{x}})}(\textbf{y} - \hat{\textbf{x}})^T)\\
&\mathcal{C}=\{\textbf{y}: \underline{d_{ij}}^2\leq f_{ij}(\textbf{y})\leq \overline{d_{ij}}^2,~~\forall (i,j)\in \mathcal{N}_x,\nonumber\\
&~~~~~~~~~~\underline{d_{ik}}^2\leq f_{ik}(\textbf{y})\leq \overline{d_{ik}}^2,~~\forall (i,k)\in \mathcal{N}_a
\}
\end{align}
 where
 \begin{small}
  \begin{align*}
  f_{ij}(\textbf{y})&=\textbf{e}_{(2i-1)(2j-1)}^T\textbf{y}\textbf{y}^T\textbf{e}_{(2i-1)(2j-1)}+\textbf{e}_{(2i)(2j)}^T\textbf{y}\textbf{y}^T\textbf{e}_{(2i)(2j)}\\
  f_{ik}(\textbf{y})&=\textbf{a}_k\textbf{a}_k^T-2{a}_{kx}\textbf{y}^T\textbf{e}_{2i-1}-2{a}_{ky}\textbf{y}^T\textbf{e}_{2i}+\textbf{e}_{2i-1}^T\textbf{y}\textbf{y}^T\textbf{e}_{2i-1}+\textbf{e}_{2i}^T\textbf{y}\textbf{y}^T\textbf{e}_{2i}
  \end{align*}
  \end{small}
  where $\textbf{e}_i\in\mathbb{R}^{2n}$ is a column vector with $1$ at the $i$th position and 0 elsewhere; $\textbf{e}_{(i)(j)}\in\mathbb{R}^{2n}$ is a column vector with $1$ at the $i$th position, $-1$ at the $j$th position, and 0 elsewhere. Three noncollinear anchors are needed to resolve translation and rotation ambiguities, but a single anchor is sufficient for establishing the boundedness of set $\mathcal{C}$.

\section{The Relaxed Estimation}\label{sec:minmax_algorithm}
Geometrically, the problem \eqref{eq:minmax_obj} is the formulation for computing the Chebyshev center of set $\mathcal{C}$. The geometrical interpretation and analysis will be given in next section.
Problem (7) is a non-convex optimization problem, for which, the convex optimization techniques cannot be directly used. In this section, we will propose a relaxed estimation algorithm. The main idea of our proposed algorithm is as follows:
\begin{enumerate}
  \item Relax the non-convex optimization problem \eqref{eq:minmax_obj} into a convex optimization problem;
  \item Change the order of optimization, which will further simplify the optimization problem;
  \item Solve the corresponding Lagrangian dual problem of the simplified problem.
\end{enumerate}
In the following, we will introduce each step in detail.

\subsection{Relaxation}
Let $\bm{\Delta}=\textbf{y}\textbf{y}^T$; then \eqref{eq:minmax_obj} can be rewritten as
 \begin{align}\label{eq:minmax_objTransfer}
\mathop {\min }\limits_{\hat{\textbf{x}}} \mathop {\max }\limits_{(\textbf{y},\bm{\Delta}) \in \mathcal{G}} \textrm{Tr}(\bm{\Delta}-2\hat{\textbf{x}}\textbf{y}^T+\hat{\textbf{x}}\hat{\textbf{x}}^T)
\end{align}
where $\mathcal{G}$ is the constraint set:

\begin{align}\label{eq:constraint_G}
&\mathcal{G}=\{(\textbf{y},\bm{\Delta}): \underline{d_{ij}}^2\leq g_{ij}(\bm{\Delta})\leq \overline{d_{ij}}^2,~~\forall (i,j)\in \mathcal{N}_x, \nonumber\\
&~~~~~~~~~~~~~~~~~\underline{d_{ik}}^2\leq g_{ik}(\textbf{y},\bm{\Delta})\leq \overline{d_{ik}}^2,~~\forall (i,k)\in \mathcal{N}_a \nonumber\\
&~~~~~~~~~~~~~~~~~\bm{\Delta}=\textbf{y}\textbf{y}^T \}
\end{align}
and
\begin{small}
\begin{align*}
g_{ij}(\bm{\Delta})&=\textbf{e}_{(2i-1)(2j-1)}^T\bm{\Delta}\textbf{e}_{(2i-1)(2j-1)}+\textbf{e}_{(2i)(2j)}^T\bm{\Delta}\textbf{e}_{(2i)(2j)}\\
g_{ik}(\textbf{y},\bm{\Delta})&=\textbf{a}_k\textbf{a}_k^T-2{a}_{kx}\textbf{y}^T\textbf{e}_{2i-1}-2{a}_{ky}\textbf{y}^T\textbf{e}_{2i}+\textbf{e}_{2i-1}^T\bm{\Delta}\textbf{e}_{2i-1}+\textbf{e}_{2i}^T\bm{\Delta}\textbf{e}_{2i}
\end{align*}
\end{small}
 The equality constraint $\bm{\Delta}=\textbf{y}\textbf{y}^T$ in \eqref{eq:constraint_G} is not affine, which makes $\mathcal{G}$ a non-convex set\cite{boyd2004convex}. The optimization problem cannot be directly solved by convex optimization methods. As commonly done in the field \cite{biswas2006semidefinite,luo2010SDR}, we make the following relaxation:
\begin{equation}\label{eq:relaxation}
\bm{\Delta} \succeq \textbf{y}\textbf{y}^T
\end{equation}
where the notation means that $\bm{\Delta}-\textbf{y}\textbf{y}^T $ is a positive semidefinite matrix. To distinguish the relaxed $\bm{\Delta}$ from the original one, we use $\bm{\Delta}_r$ to denote the relaxed $\bm{\Delta}$. Then the relaxed constraint set, which is now convex, becomes
\begin{subequations}
\begin{align}
&\mathcal{Q}=\{(\textbf{y},\bm{\Delta}_r): \underline{d_{ij}}^2\leq g_{ij}(\bm{\Delta}_r)\leq \overline{d_{ij}}^2,~~\forall (i,j)\in \mathcal{N}_x,\label{eq:g_ij_Constraint}\\
&~~~~~~~~~~~~~~~~~~~\underline{d_{ik}}^2\leq g_{ik}(\textbf{y},\bm{\Delta}_r)\leq \overline{d_{ik}}^2,~~\forall (i,k)\in \mathcal{N}_a \label{eq:g_ik_Constraint}\\
&~~~~~~~~~~~~~~~~~~~\bm{\Delta}_r\succeq\textbf{y}\textbf{y}^T\label{eq:relaxedConstraint}\}
\end{align}
\end{subequations}
Geometrically, \eqref{eq:g_ij_Constraint} and \eqref{eq:g_ik_Constraint} constitute a convex polytope, which is a closed and bounded set. Moreover, the inequality $\bm{\Delta}_r\succeq\textbf{y}\textbf{y}^T$, which is equivalent to $\left[ {\begin{smallmatrix}
\bm{\Delta}_r&{\textbf{y}}\\
{\textbf{y}^T}&1
\end{smallmatrix}} \right]\succeq0$  \cite{SDPtheory}, defines a positive semidefinite cone (closed but unbounded set)\cite{boyd2004convex}. Set $\mathcal{Q}$ is the intersection of a convex polytope and a positive semidefinite cone. Therefore, $\mathcal{Q}$ is a closed and bounded set.

The relaxed problem becomes a convex optimization problem
 \begin{align}\label{eq:minmaxRelax_objTransfer}
\mathop {\min }\limits_{\hat{\textbf{x}}} \mathop {\max }\limits_{(\textbf{y},\bm{\Delta}_r) \in \mathcal{Q}} \textrm{Tr}(\bm{\Delta}_r-2\hat{\textbf{x}}\textbf{y}^T+\hat{\textbf{x}}\hat{\textbf{x}}^T)
\end{align}

\subsection{Change of Optimization Order}
In \eqref{eq:minmaxRelax_objTransfer}, the outer minimization part is an unconstrained optimization problem, which is straightforward, while the inner maximization part is a constrained optimization problem over $(\textbf{y},\bm{\Delta}_r)$. In an effort to simplify the problem, we consider whether we can change the order of these two parts.

In our problem \eqref{eq:minmaxRelax_objTransfer}, the objective function is continuous, finite, and convex in $\hat{\textbf{x}}$. Since the objective function is linear with $\textbf{y}$ and $\bm{\Delta}_r$, it is concave in $(\textbf{y},\bm{\Delta}_r)$. Both the feasible sets of $\hat{\textbf{x}}$ and $(\textbf{y},\bm{\Delta}_r)$ are closed. Moreover, set $\mathcal{Q}$ is bounded. Consequently, according to Corollary 37.3.2 in \cite{rockafellar2015convex}, we can interchange the order of minimization and maximization. The equivalent optimization problem to \eqref{eq:minmaxRelax_objTransfer} becomes
\begin{align}\label{eq:maxmin_obj}
\mathop {\max }\limits_{(\textbf{y},\bm{\Delta}_r) \in \mathcal{Q}} \mathop {\min }\limits_{\hat{\textbf{x}}} \textrm{Tr}(\bm{\Delta}_r-2\hat{\textbf{x}}\textbf{y}^T+\hat{\textbf{x}}\hat{\textbf{x}}^T)
\end{align}

It is straightforward that the optimal solution of the inner minimization problem in \eqref{eq:maxmin_obj} is $\hat{\textbf{x}}(\textbf{y})=\textbf{y}$. Hence, the equivalent optimization problem to \eqref{eq:maxmin_obj} becomes
\begin{align}\label{eq:max_obj}
\mathop {\max }\limits_{(\textbf{y},\bm{\Delta}_r) \in \mathcal{Q}} \textrm{Tr}(\bm{\Delta}_r-\textbf{y}\textbf{y}^T)
\end{align}

In our problem formulation, the problem of minimizing the worst case error is formulated as a min-max problem. That is, first assuming a given location estimate, since any point in the feasible region can be the true
location of the sensor, we find the point in the feasible region that maximizes the difference between the location estimate and the location of that point (which can be the potential true location). This represents the worst case location estimation error. Secondly, we find the location estimate that minimizes such worst case error. Alternatively, the min-max problem can also be formulated as a max-min problem. That is, assuming the ``true" location is fixed (which can be any point in the feasible region), we first find the location estimate that minimizes the difference between the location estimate and the ``true" location, i.e., the minimum location estimation error assuming the ``true" location is fixed. We recognize that when the true location of the sensor is at a different point of the feasible region, the corresponding minimum location estimation error will be different and some points in the feasible region may deliver more accurate location estimates than some other points. This reflects the fact that other things being equal, some geometric points may be more accurately localized than some other points. Secondly, we find the ¡°true¡± location within the feasible region that delivers the worst case minimum location estimation error. It can be shown analytically that the min-max and the max-min problem are equivalent. The maximization problem in \eqref{eq:max_obj} corresponds to the last step: i.e., finding the ``true" sensor location that delivers the worst case minimum location estimation error.

\subsection{Dual Problem}
Since \eqref{eq:max_obj} is a convex optimization problem and strictly feasible, strong duality holds\cite{boyd2004convex}. Problem \eqref{eq:max_obj} can be solved through its dual problem. The Lagrangian dual function of \eqref{eq:max_obj} is
\begin{small}
\begin{align}\label{eq:Lagrangian}
L(\textbf{y},\bm{\Delta}_r,\alpha_{ij},\beta_{ij},\omega_{ik},\varphi_{ik},\bm{\lambda})&=\textrm{Tr}((\textbf{I}_{2n}+\bm{\lambda})(\bm{\Delta}_r-\textbf{y}\textbf{y}^T))+\sum\limits_{(i,j)\in\mathcal{N}_x}{\alpha_{ij}}\left({g_{ij}(\bm{\Delta}_r)-\underline{d_{ij}}^2}\right)
\nonumber\\
&+\sum\limits_{(i,j)\in\mathcal{N}_x}{\beta_{ij}}\left({-g_{ij}(\bm{\Delta}_r)+\overline{d_{ij}}}^2\right)+\sum\limits_{(i,k)\in\mathcal{N}_a}{\omega_{ik}}\left({g_{ik}(\textbf{y},\bm{\Delta}_r)-\underline{d_{ik}}^2}\right)
\nonumber\\
&+\sum\limits_{(i,k)\in\mathcal{N}_a}{\varphi_{ik}}\left({-g_{ik}(\textbf{y},\bm{\Delta}_r)+\overline{d_{ik}}^2}\right)
\end{align}
\end{small}
where $\textbf{I}_{2n}$ denotes a $2n\times2n$ identity matrix, and the dual variables are $\alpha_{ij}\in\mathbb{R}$, $\beta_{ij}\in\mathbb{R}$, $\omega_{ik}\in\mathbb{R}$, $\varphi_{ik}\in\mathbb{R}$ and $\bm{\lambda}\in\mathbb{R}^{2n\times2n}$, which obey the constraints: $\alpha_{ij}\geq0$, $\beta_{ij}\geq0$, $\omega_{ik}\geq0$, $\varphi_{ik}\geq{0}$ and $\bm{\lambda}\succeq\textbf{0}$. The dual problem is
\begin{align}\label{eq:max_Lagrangian}
\mathop {\min }\limits_{({\alpha _{ij}},{\beta _{ij}},{\omega _{ik}},{\varphi _{ik}},{\bm{\lambda}})}\mathop{\sup}\limits_{(\textbf{y},\bm{\Delta}_r)} L(\textbf{y},\bm{\Delta}_r,\alpha_{ij},\beta_{ij},\omega_{ik},\varphi_{ik},\bm{\lambda})
\end{align}
To simplify the notation, let $L$ denote $L(\textbf{y},\bm{\Delta}_r,\alpha_{ij},\beta_{ij},\omega_{ik},\varphi_{ik},\bm{\lambda})$. The inner maximization problem can be solved by letting the derivative of $L$ with respect to $\textbf{y}$ and $\bm{\Delta}_r$ equal to $\textbf{0}$, i.e.,
\begin{align}\label{eq:derivative_L_Y}
\frac{{\partial {L}}}{{\partial \textbf{y}}} = \textbf{0} ~~~~\frac{{\partial {{L}}}}{{\partial \bm{\Delta}_r }}=\textbf{0}
\end{align}
From \eqref{eq:derivative_L_Y}, the optimal value of $\textbf{y}$ satisfies
\begin{align}\label{eq:Y_optimal}
\hat{\textbf{y}}=-(\textbf{I}_{2n}+\bm{\lambda})^{-1}\sum\limits_{(i,k) \in {\mathcal{N}_a}} ({\omega_{ik}}-{\varphi_{ik}}) ({{a}_{kx}\textbf{e}_{2i-1}}+{a}_{ky}\textbf{e}_{2i})
\end{align}
and
\begin{align}\label{eq:I+lamda}
\textbf{I}_{2n}+\bm{\lambda}  =-\sum\limits_{{(i,j)} \in {\mathcal{N}_x}} (\alpha _{ij}-\beta _{ij})\textbf{E}_{ij}- \sum\limits_{(i,k) \in {\mathcal{N}_a}} (\omega _{ik}-{\varphi _{ik}})\textbf{E}_i
\end{align}
where
%\begin{small}
\begin{align*}
\textbf{E}_{ij}&={\textbf{e}_{(2i-1)(2j-1)}\textbf{e}^T_{(2i-1)(2j-1)}+\textbf{e}_{(2i)(2j)}\textbf{e}^T_{(2i)(2j)}}\\
\textbf{E}_{i}&={{\textbf{e}_{2i-1}}\textbf{e}_{2i-1}^T+{\textbf{e}_{2i}\textbf{e}_{2i}^T}}
\end{align*}
%\end{small}
By substituting \eqref{eq:Y_optimal} and \eqref{eq:I+lamda} into \eqref{eq:Lagrangian}, the dual function can be obtained as
\begin{align}\label{eq:dual function}
{g}(\alpha_{ij},\beta_{ij},\omega_{ik},\varphi_{ik})=\textbf{f}^T(\textbf{I}_{2n}+\bm{\lambda})^{-1}\textbf{f}+h({\alpha _{ij}},{\beta _{ij}},{\omega _{ik}},{\varphi _{ik}})
\end{align}
where
$$\textbf{f}= \sum\limits_{(i,k) \in {{\cal N}_a}} ( {\omega _{ik}} - {\varphi _{ik}})(a_{kx}\textbf{e}_{2i - 1} + a_{ky}\textbf{e}_{2i})$$
\begin{small}
\begin{align*}
h({\alpha _{ij}},{\beta _{ij}},{\omega _{ik}},{\varphi _{ik}})=&-\sum\limits_{(i,j)\in\mathcal{N}_x}\alpha_{ij}\underline{d_{ij}}^2+\sum\limits_{(i,j)\in\mathcal{N}_x}\beta_{ij}\overline{d_{ij}}^2\\
&+\sum\limits_{(i,k)\in\mathcal{N}_a}{\omega_{ik}}(\textbf{a}_k\textbf{a}_k^T-\underline{d_{ik}}^2)+\sum\limits_{(i,k)\in\mathcal{N}_a}{\varphi_{ik}}(\overline{d_{ik}}^2-\textbf{a}_k\textbf{a}_k^T)
\end{align*}
\end{small}
The dual optimization problem becomes
\begin{subequations}
\begin{align}\label{eq:dual_Problem}
&\mathop {\min }\limits_{{\alpha _{ij}},{\beta _{ij}},{\omega _{ik}},{\varphi _{ik}},\bm{\lambda}} \textbf{f}^T(\textbf{I}_{2n}+\bm{\lambda})^{-1}\textbf{f}+h({\alpha _{ij}},{\beta _{ij}},{\omega _{ik}},{\varphi _{ik}})\\
&s.t. ~~~{\alpha _{ij}}, {\beta _{ij}}, {\omega _{ik}}, {\varphi _{ik}} \ge 0,~\bm{\lambda}\succeq \textbf{0},\label{subeq:dualVariable_central}\\
 &~~~~~~~~ \forall (i,j) \in {\mathcal{N}_x}, \forall (i,k) \in {\mathcal{N}_a}\label{subeq:ijIndex}
\end{align}
\end{subequations}
Assume that $t$ is a scalar such that $\textbf{f}^T(\textbf{I}_{2n}+\bm{\lambda})^{-1}\textbf{f}\leq t$. From the property of Schur complement, we have
\begin{small}
\begin{align}\label{eq:shcur_Central}
\left[ {\begin{array}{*{20}{c}}
\textbf{I}_{2n}+\bm{\lambda}&{\textbf{f}}\\
{\textbf{f}^T}&t
\end{array}} \right]\succeq0
\end{align}
\end{small}
Problem \eqref{eq:dual_Problem} can be transformed into an SDP problem£º
\begin{align}\label{eq:SDP}
&\mathop {\min }\limits_{{\alpha _{ij}},{\beta _{ij}},{\omega _{ik}},{\varphi _{ik}},\bm{\lambda}} {t} + h({\alpha _{ij}},{\beta _{ij}},{\omega _{ik}},{\varphi _{ik}}) \nonumber\\
&s.t.~~~ \eqref{eq:shcur_Central}\eqref{subeq:dualVariable_central}\eqref{subeq:ijIndex}
\end{align}
The above SDP can be numerically and efficiently solved to any arbitrary accuracy by many existing SDP solvers and toolboxes, e.g., SeDuMi \cite{SeDuMi_Manual}, CVX\cite{grant2008cvx,luo2010SDR}.

Let ${\hat{\alpha }_{ij}},{\hat{\beta} _{ij}},{\hat{\omega }_{ik}},{\hat{\varphi} _{ik}},\hat{\bm{\lambda}}$ denote the solution of \eqref{eq:SDP}; then the estimate of $\textbf{x}$ is
\begin{equation}\label{eq:final_estimate}
{\textbf{x}}_{est}=-(\textbf{I}_{2n}+\hat{\bm{\lambda}})^{-1}\sum\limits_{(i,k) \in {\mathcal{N}_a}} ({\hat{\omega}_{ik}}-{\hat{\varphi}_{ik}}) ({{a}_{kx}\textbf{e}_{2i-1}}+{a}_{ky}\textbf{e}_{2i})
\end{equation}

In the above SDP, there are $2|\mathcal{N}_x|+2|\mathcal{N}_a|+n(2n+1)+1$ scalar variables to be optimized, and the number of scalar equality/inequality constraints is $2|\mathcal{N}_x|+2|\mathcal{N}_a|$, the number of linear matrix inequality (LMI) constraints is $2$, and the size of the LMI is at most $(2n+1)\times (2n+1)$. In \eqref{eq:final_estimate}, the computational complexity of SDP is $\mathcal{{O}}(n^6)$ \cite{boyd2004convex} and the computational complexity of the matrix inverse operation is $\mathcal{{O}}(n^3)$. We can conclude that the computational complexity of our estimation algorithm is $\mathcal{{O}}(n^6)$.

\section{Geometrical Interpretation}\label{sec:geometrical interpretation}
In this section, we will give a geometric interpretation of our problem. Firstly, we will show that the original problem \eqref{eq:minmax_obj} is a standard one, of finding the Chebyshev center of $\mathcal{C}$\cite{eldar2008minimax,wu2013Automatica_ACC,uniqueCenter}. Secondly, in section \ref{sec:minmax_algorithm}, we make a relaxation of the original problem, the solution of which may be thought of as a relaxed Chebyshev center. The detailed interpretations are presented in the following subsections.
\subsection{Chebyshev Center}
 \begin{figure}[!ht]
  \centering
  \subfigure[$\textbf{x}_{Cheby}$ is feasible.]{
   \begin{minipage}[t]{0.3\linewidth}
 \includegraphics[width=\textwidth]{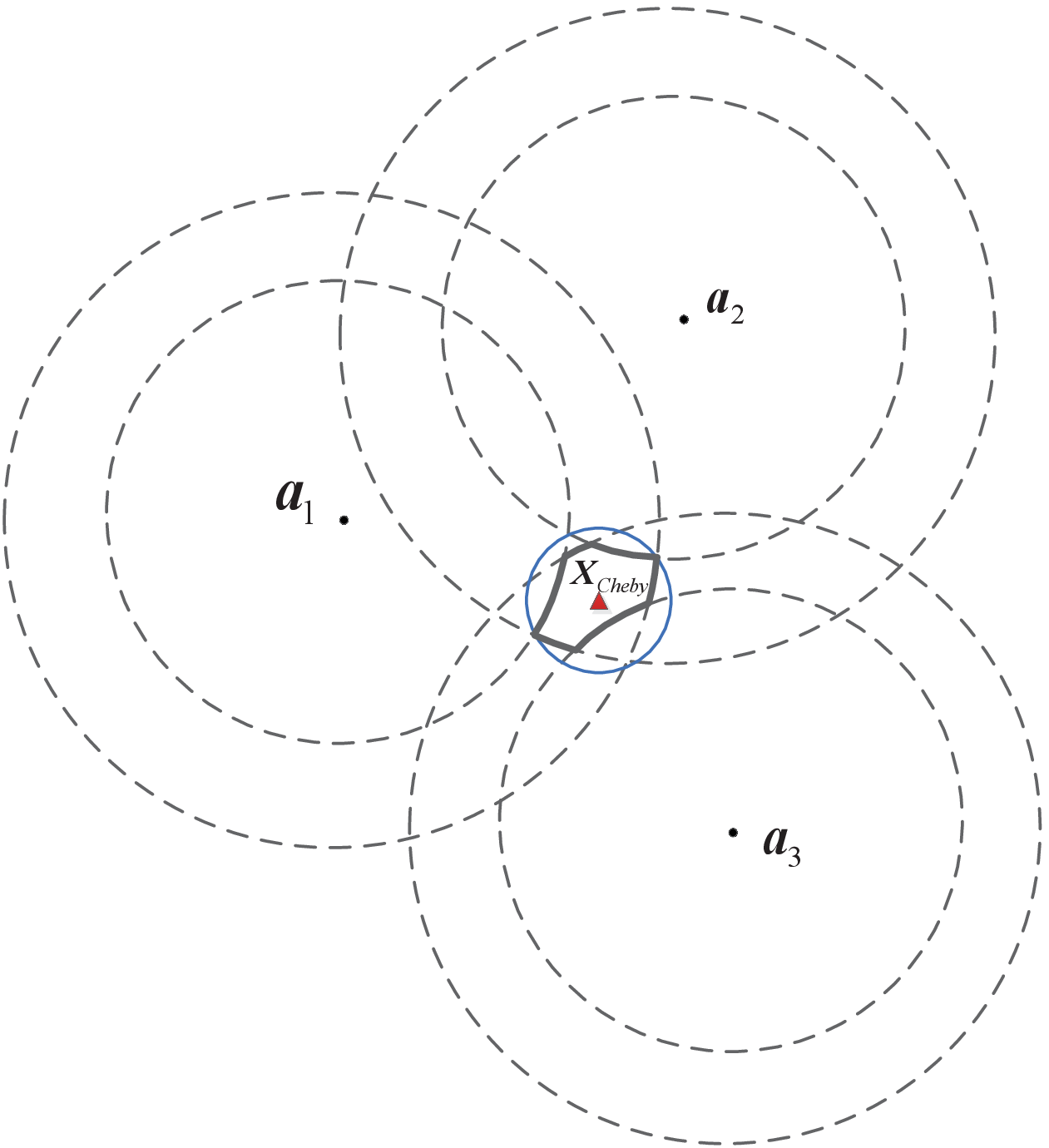}
\label{subfig:ChebyshevCenter_feasible}
\vspace{-15pt}
\end{minipage}}
 \hspace{15pt}
  \subfigure[$\textbf{x}_{Cheby}$ is infeasible.]{
   \begin{minipage}[t]{0.3\linewidth}
 \includegraphics[width=\textwidth]{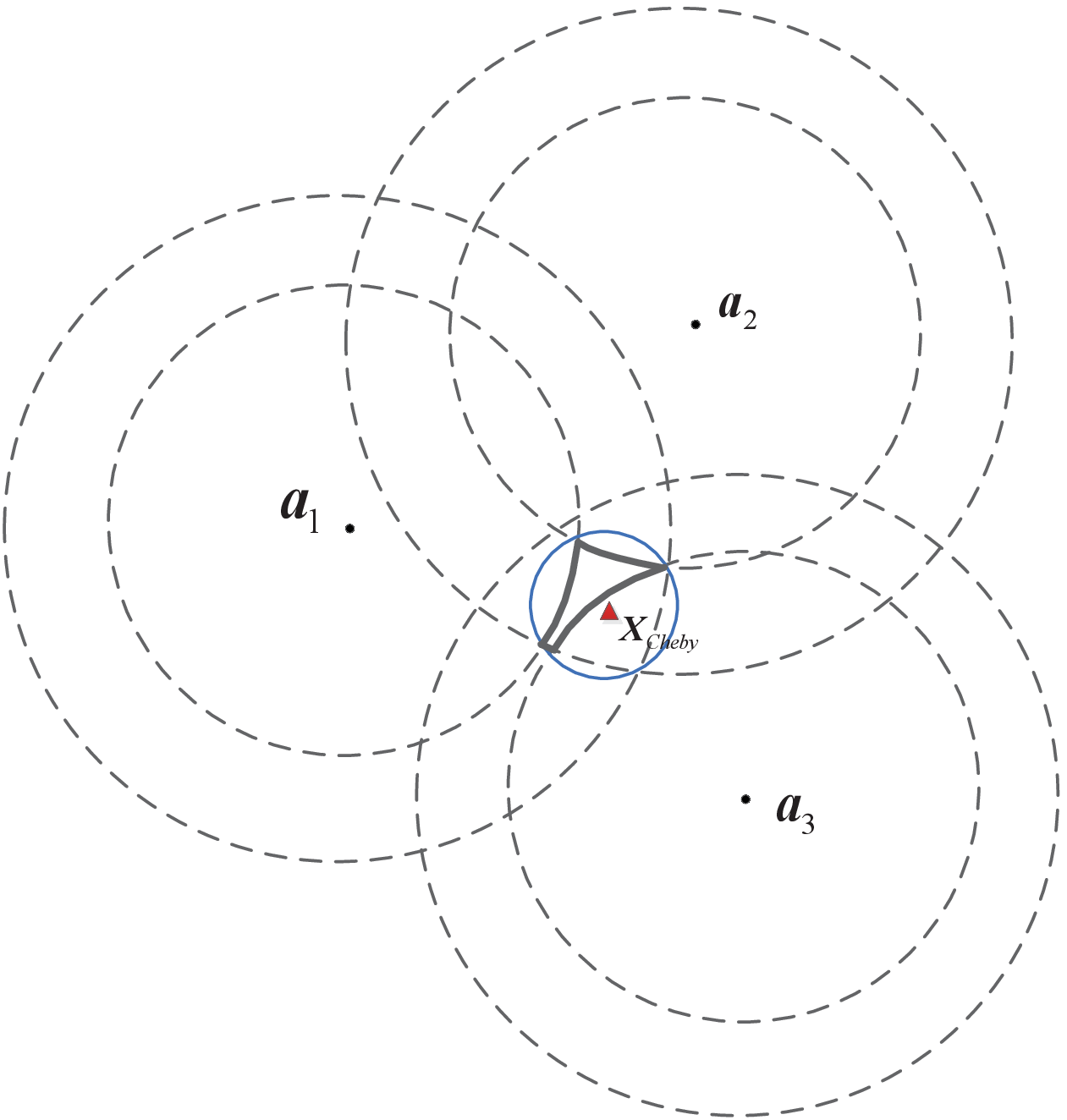}
\label{subfig:ChebyshevCenter_infeasible}
\vspace{-15pt}
\end{minipage}}
\caption{Chebyshev center of a non-convex set, which is the intersection of three feasible sets of a single sensor.}
\label{fig:ChebyshevCenter_nonconvex}
\vspace{-20pt}
 \end{figure}

Geometrically, our objective is to find the Chebyshev center of $\mathcal{C}$, i.e., the center of the minimum ball enclosing $\mathcal{C}$ \cite{eldar2008minimax}. The equivalent problem to \eqref{eq:minmax_obj} is
  \begin{align}\label{eq:equalChebyProblem}
  \mathop {\min }\limits_{\hat{\textbf{x}}} \left\{R_c:{\parallel\textbf{y}-\hat{\textbf{x}}\parallel}^2\leq R_c^2, \forall\textbf{y}\in\mathcal{C} \right\}
  \end{align}
  In our problem, set $\mathcal{C}$, the feasible set of the position vector $\textbf{y}$, is non-convex. This means that the Chebyshev center of $\mathcal{C}$ may not lie in $\mathcal{C}$. We take a single sensor localization problem as an example to illustrate that. As shown in Fig.~\ref{fig:ChebyshevCenter_nonconvex}, $\textbf{a}_1$, $\textbf{a}_2$ and $\textbf{a}_3$ are anchors. From their range measurements with bounded errors, we can determine a non-convex feasible set for the target sensor, which is the region surrounded by the solid bold curve. The minimum circle enclosing the feasible set is the smallest solid circle, and its center (the triangle) is the Chebyshev center $\textbf{x}_{Cheby}$. In Fig.~\ref{subfig:ChebyshevCenter_feasible}, ${\textbf{x}}_{Cheby}$ lies in the feasible set, but in Fig.~\ref{subfig:ChebyshevCenter_infeasible}, it is obvious that ${\textbf{x}}_{Cheby}$ does not lie in the feasible set. Even though we cannot be sure whether ${\textbf{x}}_{Cheby}$ lies in the feasible set, if we take ${\textbf{x}}_{Cheby}$ as the position estimate, we will get the minimum worst-case estimation error.

In practice, sometimes a feasible estimate is preferred over an estimate with the minimum estimation error bound when the latter is outside the feasible set. If the Chebyshev center is infeasible, as shown in Fig.~\ref{subfig:ChebyshevCenter_infeasible}, one general way to proceed is to take the projection of $\textbf{x}_{Cheby}$ onto the feasible set as the estimate. The following proposition compares the estimation error bound obtained with the Chebyshev center and that obtained with its projection onto the feasible set.

\begin{Proposition}\label{proposition2:projection bound}
Suppose $\mathcal{C}\subset\mathbb{R}^{2n}$ is a closed set, let $\textbf{x}_{Cheby}\in\mathbb{R}^{2n}$ be the Chebyshev center of $\mathcal{C}$ and suppose $\textbf{x}_{Cheby}\notin\mathcal{C}$. Let $\textbf{x}_{p}=\mathcal{P}_{\mathcal{C}}(\textbf{x}_{Cheby})=\arg\mathop {\min }\limits_{\textbf{y}\in \mathcal{C}}{\parallel\textbf{x}_{Cheby}-\textbf{y}\parallel}$, where $\mathcal{P}_{\mathcal{C}}(\cdot)$ denotes the projection operator onto set $\mathcal{C}$ \cite{boyd2004convex}. Let $R_c=\mathop {\max }\limits_{\textbf{y}\in \mathcal{C}}{\parallel{\textbf{x}_{Cheby}-\textbf{y}\parallel}}$ and $R_p=\mathop {\max }\limits_{\textbf{y}\in \mathcal{C}}{\parallel\textbf{x}_{p}-\textbf{y}\parallel}$. Then
\begin{align}\label{eq:projectionBoundVSchebyBound}
R_c \leq R_p\leq 2R_c
\end{align}
\end{Proposition}
\begin{proof}
From the definition of Chebyshev center, we can easily obtain
\begin{align}
R_c^2=\mathop{\max}\limits_{\textbf{y}\in\mathcal{C}}{\parallel{\textbf{x}_{Cheby}-\textbf{y}\parallel}^2}=\mathop{\min}\limits_{\hat{\textbf{x}}}\mathop {\max}\limits_{\textbf{y}\in\mathcal{C}}{\parallel{\hat{\textbf{x}}-\textbf{y}\parallel}^2}
\leq \mathop{\max}\limits_{\textbf{y}\in\mathcal{C}}{\parallel{\textbf{x}_{p}-\textbf{y}\parallel}^2}=R_p^2
\end{align}
Apparently, $R_c \leq R_p$. We also note that
\begin{align}
\mathop{\max}\limits_{\textbf{y}\in\mathcal{C}}{\parallel{\textbf{x}_{p}-\textbf{y}\parallel}}
&=\mathop{\max}\limits_{\textbf{y}\in\mathcal{C}}{\parallel{\textbf{x}_{p}-\textbf{x}_{Cheby}+\textbf{x}_{Cheby}-\textbf{y}\parallel}}\nonumber\\
&\leq\mathop{\max}\limits_{\textbf{y}\in\mathcal{C}}(\parallel{\textbf{x}_{p}-\textbf{x}_{Cheby}\parallel+\parallel\textbf{x}_{Cheby}-\textbf{y}\parallel})\nonumber\\
&\leq \parallel{\textbf{x}_{p}-\textbf{x}_{Cheby}\parallel+\mathop{\max}\limits_{\textbf{y}\in\mathcal{C}}\parallel\textbf{x}_{Cheby}-\textbf{y}\parallel}
\end{align}
Since $\textbf{x}_{p}\in\mathcal{C}$, then $\parallel\textbf{x}_{p}-\textbf{x}_{Cheby}\parallel\leq \mathop{\max}\limits_{\textbf{y}\in\mathcal{C}}{\parallel{\textbf{x}_{Cheby}-\textbf{y}\parallel}}$. Thus we have
\begin{align}
\mathop{\max}\limits_{\textbf{y}\in\mathcal{C}}{\parallel{\textbf{x}_{p}-\textbf{y}\parallel}}\leq 2\mathop{\max}\limits_{\textbf{y}\in\mathcal{C}}\parallel\textbf{x}_{Cheby}-\textbf{y}\parallel
\end{align}
Equivalently, $R_p \leq 2R_c$. The proof is complete.
\end{proof}

From Proposition \ref{proposition2:projection bound}, we can see that $\textbf{x}_p$ is feasible, but the upper bound of the estimation error of $\textbf{x}_p$ is larger than that of $\textbf{x}_{Cheby}$. In practice, we cannot just say one estimate is definitely better than another; such a statement depends on the metric used to evaluate localization performance. If we need a feasible estimate with a small estimation error bound, $\textbf{x}_p$ would be a good choice. Nevertheless, if we need an estimate with the minimum estimation error bound, $\textbf{x}_{Cheby}$ is obviously a better choice. In this paper, we only take the estimation error bound as the metric evaluating localization performance, which makes $\textbf{x}_{Cheby}$ the best choice.

\subsection{Relaxed Chebyshev Center}
Finding the Chebyshev center is an NP-hard problem except in some special cases \cite{wu2013Automatica_ACC}. In our problem, because the constraint set is not convex, it is difficult to obtain $\textbf{x}_{Cheby}$. In \eqref{eq:final_estimate}, ${\textbf{x}}_{est}$ is a relaxed estimate of the Chebyshev center. It is straightforward that the estimation error bound of ${\textbf{x}}_{est}$ is no less than that for $\textbf{x}_{Cheby}$,
 \begin{align}\label{eq:Xest VS Xcheby}
 \mathop{\max}\limits_{\textbf{y}\in\mathcal{C}}{\parallel{\textbf{x}_{Cheby}-\textbf{y}\parallel}^2}=\mathop{\min}\limits_{\hat{\textbf{x}}}\mathop {\max}\limits_{\textbf{y}\in\mathcal{C}}{\parallel{\hat{\textbf{x}}-\textbf{y}\parallel}^2}
\leq \mathop{\max}\limits_{\textbf{y}\in\mathcal{C}}{\parallel{\textbf{x}_{est}-\textbf{y}\parallel}^2}
 \end{align}

 Even though we cannot be sure whether ${\textbf{x}}_{est}$ lies in the feasible set $\mathcal{C}$, when we minimize the worst-case estimation error over a relaxed constraint set, the estimation error satisfies
\begin{align}\label{eq:error_bound}
\parallel{\textbf{x}}_{est}-\textbf{x}\parallel^2
&\leq \mathop {\max }\limits_{\textbf{y}\in \mathcal{C}}{\parallel{\textbf{x}}_{est}-\textbf{y}\parallel}^2\nonumber\\
&=\mathop {\max }\limits_{(\textbf{y},\bm{\Delta}) \in \mathcal{G}} \textrm{Tr}(\bm{\Delta}-2{\textbf{x}}_{est}\textbf{y}^T+{\textbf{x}}_{est}{\textbf{x}}_{est}^T)\nonumber\\
&\leq\mathop {\max }\limits_{(\textbf{y},\bm{\Delta}_r) \in \mathcal{Q}} \textrm{Tr}(\bm{\Delta}_r-2{\textbf{x}}_{est}\textbf{y}^T+{\textbf{x}}_{est}{\textbf{x}}_{est}^T)\nonumber\\
&=\mathop {\min }\limits_{\hat{\textbf{x}}} \mathop {\max }\limits_{(\textbf{y},\bm{\Delta}_r) \in \mathcal{Q}} \textrm{Tr}(\bm{\Delta}_r-2\hat{\textbf{x}}\textbf{y}^T+\hat{\textbf{x}}\hat{\textbf{x}}^T)
\end{align}
Then we can say the estimation error is upper bounded by the optimal value of \eqref{eq:minmaxRelax_objTransfer}.

In Euclidean space, every closed convex bounded set has a unique Chebyshev centre\cite{uniqueCenter}. However, our original constraint set $\mathcal{C}$ is not convex, and the Chebyshev center $\textbf{x}_{Cheby}$ may not be unique. Through relaxation, the relaxed set $\mathcal{Q}$ becomes a closed convex bounded set. Therefore, we can obtain a unique Chebyshev center $\textbf{x}_{est}$ of the relaxed set $\mathcal{Q}$.

\section{Distributed Estimation} \label{sec:distributed alg}
In practice, the network scale is often very large, including hundreds, even thousands of sensors and sometimes only a small number of anchors. If the network is localized in a centralized way, it may result in an extremely high communication burden and computational complexity at the central processor. Therefore, in this section, we consider an implementation of the above relaxed estimation method in a distributed way. In the distributed algorithm, each sensor will first make an initial guess about its position, which is denoted by $\hat{\textbf{x}}_i(0), i=1,2,\cdots,n$. Then, each sensor takes its neighbor nodes as `anchors', and iteratively estimates its position.
\subsection{Initial Estimation}
\subsubsection{Sensor-Anchor Distance Estimation}
Since the sensing range of each sensor is limited, not all the sensors have direct connections with anchors. To make an initial estimation on each sensor's position, we first estimate the Euclidean distance between each sensor and anchor through information exchange between neighboring nodes.

Similarly to the DV-hop localization algorithm \cite{DVhop2003}, each anchor broadcasts its position with a hop-counter initialized to one through the network. Let $n_{ik}$ denote the hop-count value from anchor $k$ to sensor $i$. Sensor $i$ would compute the lower and the upper bounds of $\parallel {{\textbf{x}_i} - {\textbf{a}_k}} \parallel$, and broadcast these bounds along with ${\textbf{a}_k}$ and $n_{ik}$ to its neighbors. The procedure of computing the lower and the upper bounds is explained in next paragraph. Each receiving sensor will select the information along the shortest path, measured by the number of hops, to compute the bound of distance between this sensor and anchor $k$.

If $(i,k)\in\mathcal{N}_a$, i.e., sensor $i$ has a direct connection with anchor $k$, i.e., $n_{ik}=1$, the true distance between $i$ and anchor $k$ satisfies \eqref{eq:feasible set(b)}. If $(i,k)\notin\mathcal{N}_a$, i.e., sensor $i$ does not have a direct connection with anchor $k$, sensor $i$ will estimate $\underline{d_{ik}}$ and $\overline{d_{ik}}$ through the information received from its neighbors $\mathcal{N}_i$. For  $j\in\mathcal{N}_i$, if $j=\arg\mathop {\min}\limits_{l\in\mathcal{N}_i}n_{lk}$, sensor $i$ would use the information from sensor $j$ to estimate $\underline{d_{ik}}$ and $\overline{d_{ik}}$. The true distance between $i$ and $j$ satisfies \eqref{eq:feasible set(a)}.
Obviously,
\begin{small}
\begin{align*}
\parallel {{\textbf{x}_i} - {\textbf{x}_j}} \parallel+\parallel {{\textbf{x}_j} - {\textbf{a}_k}} \parallel\leq\overline{d_{ij}}+\overline{d_{jk}}\\
\parallel {{\textbf{x}_i} - {\textbf{x}_j}} \parallel-\parallel {{\textbf{x}_j} - {\textbf{a}_k}} \parallel\geq\underline{d_{ij}}-\overline{d_{jk}}\\
\parallel {{\textbf{x}_j} - {\textbf{a}_k}}\parallel-\parallel {{\textbf{x}_i} - {\textbf{x}_j}} \parallel\geq\underline{d_{jk}}-\overline{d_{ij}}
\end{align*}
\end{small}
Since $
|\parallel {{\textbf{x}_i} - {\textbf{x}_j}} \parallel-\parallel {{\textbf{x}_j} - {\textbf{a}_k}} \parallel|\leq\parallel {{{\textbf{x}_i} - {\textbf{x}_j}}+{\textbf{x}_j} - {\textbf{a}_k}} \parallel\leq\parallel {{\textbf{x}_i} - {\textbf{x}_j}} \parallel+\parallel {{\textbf{x}_j} - {\textbf{a}_k}} \parallel
$, we have $\underline{d_{ik}}=\max\{\underline{d_{ij}}-\overline{d_{jk}},\underline{d_{jk}}-\overline{d_{ij}},0\} $, $\overline{d_{ik}}=\overline{d_{ij}}+\overline{d_{jk}}$ and $n_{ik}=n_{jk}+1$.

By recursion, it can be shown that, the distance bounds between sensor $i$ and anchor $k$ are
\begin{align}\label{eq:dv_lowerbound}
\underline {{d_{ik}}}  = \left\{ {\begin{array}{*{20}{c}}
{\max \{{z_{ik}} - \gamma, 0\},~~~~~~~~~~~~~~~(i,k) \in {\mathcal{N}_a}}\\
{\max \{ \underline {{d_{ij}}}  - \overline {{d_{jk}}} ,\underline {{d_{jk}}} - \overline {{d_{ij}}} ,0\} ,~~(i,k) \notin {\mathcal{N}_a}}
\end{array}} \right.
\end{align}
\begin{align}\label{eq:dv_upperbound}
\overline{d_{ik}} = \left\{ \begin{array}{*{20}{c}}
{{z_{ik}}+ \gamma,~~~~~~~~ (i,k) \in {\mathcal{N}_a}}\\
{\overline{d_{ij}}+\overline{d_{jk}}, ~~~~~~(i,k) \notin {\mathcal{N}_a},}
\end{array} \right.
\end{align}
where $j = \arg \mathop {\min }\limits_{l \in {{\cal N}_i}} {n_{lk}}$.

\begin{remark}
If the network is very large, the initial estimates of sensor-anchor distance bounds are obtained through an $n_k$-hop path, where $n_k$ might be very large. Consequently, the estimation performance will be degraded. In this case, some approximation strategies can be applied to reduce $n_k$ by dividing such large-scale network into smaller subnetworks. Many existing techniques are available for network division, e.g., the approach applied in \cite{biswas2006SDP_TOSN}, clustering in \cite{ICASSP2015}, clique extraction in \cite{SPCOM16}, etc.
\end{remark}

\subsubsection{Initial Position Estimation}\label{subsec:initial estimate}
After we have obtained the distance bounds between each sensor-anchor pair, the initial estimate of sensor $i$'s position is obtained by solving the following optimization problem
\begin{align}\label{eq:singleNode_intial}
&\mathop {\min }\limits_{{\hat {\textbf{x}}_i}} \mathop {\max }\limits_{{\textbf{y}_i}} {{{\parallel {{\textbf{y}_i} - {\hat {\textbf{x}}_i}} \parallel}^2}} \nonumber\\
&s.t. ~~~~~~ \underline{d_{ik}}^2 \le \parallel {{\textbf{y}_i} - {\textbf{a}_k}}\parallel^2 \le\overline{d_{ik}}^2 ,~~\forall k \in {{\cal V}_a}
\end{align}
Let $\Delta_i=\textbf{y}_i\textbf{y}_i^T$; an equivalent problem to \eqref{eq:singleNode_intial} is
\begin{align}\label{eq:singleNode_original2}
&\mathop {\min }\limits_{{\hat {\textbf{x}}_i}} \mathop {\max }\limits_{({\textbf{y}_i},\Delta_i)\in\mathcal{G}_i} \{\Delta_i-2\textbf{y}_i\hat{\textbf{x}}_i^T+\hat{\textbf{x}}_i\hat{\textbf{x}}_i^T\} \nonumber\\
&\mathcal{G}_i=\{({\textbf{y}_i},\Delta_i): \underline{d_{ik}}^2 \le \Delta_i-2\textbf{y}_i{\textbf{a}_k}^T+{\textbf{a}_k}{\textbf{a}_k^T} \le\overline{d_{ik}}^2 ,~~\forall k \in {{\cal V}_a} \nonumber\\
&~~~~~~~~~~~~~~~~~~~~\Delta_i=\textbf{y}_i\textbf{y}_i^T\}
\end{align}
The constraint set $\mathcal{G}_i$ is non-convex, as before, we relax $\Delta_i=\textbf{y}_i\textbf{y}_i^T$ into $\Delta_i\geq\textbf{y}_i\textbf{y}_i^T$. Let $\Delta_{ir}$ denote the relaxed $\Delta_i$; then problem \eqref{eq:singleNode_original2} becomes a convex optimization problem
\begin{align}\label{eq:singleNode_relaxed}
&\mathop {\min }\limits_{{\hat {\textbf{x}}_i}} \mathop {\max }\limits_{({\textbf{y}_i},\Delta_{ir})\in\mathcal{Q}_i} \{\Delta_{ir}-2\textbf{y}_i\hat{\textbf{x}}_i^T+\hat{\textbf{x}}_i\hat{\textbf{x}}_i^T\} \nonumber\\
&\mathcal{Q}_i=\{({\textbf{y}_i},\Delta_{ir}): \underline{d_{ik}}^2 \le \Delta_{ir}-2\textbf{y}_i{\textbf{a}_k}^T+{\textbf{a}_k}{\textbf{a}_k^T} \le\overline{d_{ik}}^2 ,~~\forall k \in {{\cal V}_a} \nonumber\\
&~~~~~~~~~~~~~~~~~~~\Delta_{ir}\geq\textbf{y}_i\textbf{y}_i^T\}
\end{align}

Similarly to the centralized estimation, we change the order of optimization and the solution of the minimization part is $\hat{\textbf{x}}_i=\textbf{y}_i$. Problem \eqref{eq:singleNode_relaxed} is simplified as
\begin{align}\label{eq:singleNode_replace}
 \mathop {\max }\limits_{({\textbf{y}_i},\Delta_{ir})\in\mathcal{Q}_i} \{\Delta_{ir}-\textbf{y}_i\textbf{y}_i^T\}
\end{align}
Its Lagrangian dual problem is also an SDP problem:
\begin{align}\label{eq:SDP_initial}
&\mathop {\min }\limits_{{\omega _k},{\varphi _k}} t + \sum\limits_{k \in {{\cal V}_a} } {{\omega _k}} ({\textbf{a}_k}{\textbf{a}_k^T} - {\underline {{d_{ik}}} ^2}) + \sum\limits_{k \in {{\cal V}_a} } {{\varphi _k}} ({\overline {{d_{ik}}} ^2} - {\textbf{a}_k}{\textbf{a}_k^T})\nonumber\\
&s.t.~~\left[ {\begin{smallmatrix}
&{ - \sum\limits_{k \in {{\cal V}_a}} {({\omega _k} - {\varphi _k}){I_2}} }&{\sum\limits_{k \in {{\cal V}_a}} {({\omega _k} - {\varphi _k}){\textbf{a}_k^T}} }\nonumber\\
&{\sum\limits_{k \in {{\cal V}_a}} {({\omega _k} - {\varphi _k}){\textbf{a}_k}} }&t
\end{smallmatrix}} \right]\succeq0\nonumber\\
&~~~~~~~ - \sum\limits_{k \in {{\cal V}_a}} {({\omega _k} - {\varphi _k})}  \ge 1 \nonumber\\
&~~~~~~~~{\omega _k},~{\varphi _k} \ge 0,~k \in {{\cal V}_a}
\end{align}
where ${ \omega }_{k}$ and ${{ \varphi }_{k}}$ are the dual variables. By setting the derivative of the Lagrangian function with respect to $\textbf{y}_i$ equal to $0$, the estimate of $\textbf{y}_i$ becomes a function of $\omega_k$ and $\varphi_k$, whose optimal values ${\hat \omega }_{k}$ and ${{\hat \varphi }_{k}}$ are obtained by solving \eqref{eq:SDP_initial}.
Then, the initial position estimate of sensor $i$ is
\begin{equation}\label{eq:x_estimate_initial}
\hat{\textbf{x}}_i(0)=\frac{{\sum\limits_{k \in {{\cal V}_a}} {({{\hat \omega }_{k}} - {{\hat \varphi }_{k}}){\textbf{a}_k}} }}{{\sum\limits_{k \in {{\cal V}_a}} {({{\hat \omega }_{k}} - {{\hat \varphi }_{k}})} }}
\end{equation}

From \eqref{eq:error_bound},
\begin{align}\label{eq:initial error bound}
\parallel \textbf{x}_i-\hat{\textbf{x}}_i(0)\parallel^2\leq \mathop {\min }\limits_{\hat{\textbf{x}}_i}\mathop {\max }\limits_{({\textbf{y}_i},\Delta_{ir})\in\mathcal{Q}_i} \{\Delta_{ir}-2\textbf{y}_i{\hat{\textbf{x}}_i}^T+{\hat{\textbf{x}}_i}{\hat{\textbf{x}}_i}^T\}=R_i^2(0)
\end{align}
Since strong duality holds, the value of $R_i^2(0)$ equals to the optimal value of \eqref{eq:SDP_initial}.

\subsection{Iterative Estimation}
After initial estimation, in the following iterative estimation algorithm, each sensor takes its neighbor nodes as `anchors', and iteratively estimates its position utilizing its neighbor nodes' position estimates and the corresponding range measurements. Suppose sensor $i$'s position estimate at the $\tau$-th iteration is $\hat{\textbf{x}}_i(\tau)$, and $\parallel\textbf{x}_i-{\hat {\textbf{x}}_i}(\tau)\parallel^2\leq R_i(\tau)^2$. Then the updated position at the $(\tau+1)$-th iteration can be obtained by solving the following optimization problem
 \begin{align}\label{eq:singleNode_original}
\mathop {\min }\limits_{{\hat {\textbf{x}}_i}} \mathop {\max }\limits_{{\textbf{y}_i}\in\mathcal{C}_i(\tau)} {{{\parallel {{\textbf{y}_i} - {\hat {\textbf{x}}_i}} \parallel}^2}}
\end{align}
where $\textbf{y}_i$ denotes the feasible value of $\textbf{x}_i$ in $\mathcal{C}_i(\tau)$ and
\begin{align}
\mathcal{C}_i(\tau)=\{\textbf{y}_i:&\parallel\textbf{y}_i-{\hat {\textbf{x}}_i}(\tau)\parallel^2\leq R_i(\tau)^2 \nonumber\\
&\underline{d_{ij}}^2 \le \parallel {{\textbf{y}_i} - {\hat{\textbf{x}}_j(\tau)}} \parallel^2 \le\overline{d_{ij}}^2 ,~~\forall j \in {{\cal N}_i} \}
\end{align}
Let $\Delta_i=\textbf{y}_i\textbf{y}_i^T$; problem \eqref{eq:singleNode_original} can be rewritten as
\begin{align}
&\mathop {\min }\limits_{{\hat {\textbf{x}}_i}} \mathop {\max }\limits_{({\textbf{y}_i},\Delta_i)\in\mathcal{G}_i(\tau)} \{\Delta_i-2\textbf{y}_i\hat{\textbf{x}}_i^T+\hat{\textbf{x}}_i\hat{\textbf{x}}_i^T\} \nonumber\\
&\mathcal{G}_i(\tau)=\{({\textbf{y}_i},\Delta_i): \parallel\textbf{y}_i-{\hat {\textbf{x}}_i}(\tau)\parallel^2\leq R_i(\tau)^2 \label{eq:RiConstraint}\\
 &~~~~~~~~~~~~~~~~~~~~~~~\underline{d_{ij}}^2 \le \textit{g}_{ij}(\tau)\le\overline{d_{ij}}^2 ,\forall j \in {{\cal N}_i} \label{eq:dijConstraint}\\
&~~~~~~~~~~~~~~~~~~~~~~~\Delta_i=\textbf{y}_i\textbf{y}_i^T\}
\end{align}
where $\textit{g}_{ij}(\tau)=\Delta_i-2\textbf{y}_i\hat{\textbf{x}}_j(\tau)^T+\hat{\textbf{x}}_j(\tau)\hat{\textbf{x}}_j(\tau)^T$.

Through a similar process to that used in equations \eqref{eq:singleNode_relaxed},\eqref{eq:singleNode_replace},\eqref{eq:SDP_initial}, i.e., relaxation, change of optimization order, and dual problem transformation, the position update of sensor $i$ at the $(\tau+1)$-th iteration becomes
\begin{small}
\begin{equation}\label{eq:x_estimate_distribute}
\hat{\textbf{x}}_i(\tau+1)=\frac{\hat{\alpha}(\tau+1)\hat{\textbf{x}}_i(\tau)+{\sum\limits_{j \in {{\mathcal N}_i}} {({{\hat \phi }_{j}}(\tau+1) - {{\hat \psi }_{j}}(\tau+1  )){\hat{\textbf{x}}_j}(\tau )} }}{\hat{\alpha}(\tau+1)+{\sum\limits_{j \in {{\mathcal N}_i}} {({{\hat \phi }_{j}}(\tau+1) - {{\hat \psi }_{j}}(\tau +1 ))} }}
\end{equation}
\end{small}
where $\hat{\alpha}(\tau+1)$, ${\hat \phi }_{j}(\tau+1)$ and ${{\hat \psi }_{j}}(\tau+1)$ are the optimal values of $\alpha$, ${ \phi }_{j}$ and  ${{ \psi }_{j}}$ in the following SDP problem:
\begin{align}\label{eq:SDP_distributed}
&\mathop {\min }\limits_{t_i, \alpha, {\phi _j},{\psi _j}} f_{i}(t_i,\alpha, {\phi _j},{\psi _j},\tau+1)\nonumber\\
&s.t.~~\left[ {\begin{smallmatrix}
&{(\alpha- \sum\limits_{j \in {{\mathcal N}_i}} {({\phi _j} - {\psi _j})}) {I_2}}&{(\alpha- \sum\limits_{j \in {{\mathcal N}_i}} {({\phi _j} - {\psi _j})})\hat{\textbf{x}}_j(\tau)^T}\nonumber\\
&{(\alpha- \sum\limits_{j \in {{\mathcal N}_i}} {({\phi _j} - {\psi _j})})\hat{\textbf{x}}_j(\tau)}&t_i
\end{smallmatrix}} \right]\succeq0,\nonumber\\
&~~~~~~~ \alpha- \sum\limits_{j \in {{\mathcal N}_i}} ({\phi _j} - {\psi _j})  \ge 1 ,\nonumber\\
&~~~~~~~\alpha, {\phi _j}, {\psi _j}\ge 0, j \in {\mathcal{N}_i}
\end{align}
where
\begin{align*}
f_{i}(t_i,\alpha, {\phi _j},{\psi _j},\tau+1)=&t_i+ \alpha(R_i(\tau)^2-{\hat{\textbf{x}}_i}(\tau)\hat{\textbf{x}}_i(\tau)^T) \nonumber\\
&+\sum\limits_{j \in {{\mathcal N}_i}} {{\phi _j}} ({\hat{\textbf{x}}_j}(\tau)\hat{\textbf{x}}_j(\tau)^T - {\underline {{d_{ij}}} ^2}) + \sum\limits_{j \in {{\mathcal N}_i}} {{\psi _j}} ({\overline {{d_{ij}}} ^2} - {\hat{\textbf{x}}_j}(\tau)\hat{\textbf{x}}_j(\tau)^T)
\end{align*}
and $\alpha$ is the dual variable associated with the inequality \eqref{eq:RiConstraint} in the Lagrange function, $\phi_j$ and $\psi_j$ are the dual variables associated with the inequalities \eqref{eq:dijConstraint} in the Lagrange function.

Let $R_i^2(\tau+1)$ denote the upper bound of the squared position estimation error of sensor $i$ at the $(\tau+1)$-th iteration, i.e.,
\begin{align}
\parallel\textbf{x}_i-\hat{\textbf{x}}_i(\tau+1)\parallel^2\leq R_i^2(\tau+1)
\end{align}
where $R_i^2(\tau+1)$ equals to the optimal value of $f_{i}(t_i,\alpha, {\phi _j},{\psi _j},\tau+1)$.
If $\parallel R_i^2(\tau+1)-R_i^2(\tau)\parallel\leq\epsilon$, where $\epsilon$ is a very small constant, we regard the position estimate of sensor $i$ as having converged to a steady state, and mark it as `localized'. The estimation of sensor $i$'s position will be terminated. The network localization will be terminated when all the sensors are `localized'.

The procedures of the distributed estimation including initial estimation are illustrated in Algorithm 1.

\begin{algorithm}\label{alg:distributed}
\caption{: Distributed Algorithm} {\small{
\begin{algorithmic}[1]
\FOR {$i=1$ to $n$}
\STATE \textbf{Initial distance bound estimation}:
\STATE Compute $\underline {{d_{ik}}}$ and ${\overline{{d_{ik}}}}$, $\forall k \in \mathcal{V}_a$, as \eqref{eq:dv_lowerbound} and \eqref{eq:dv_upperbound}.
\STATE \textbf{Initial position estimation}:
\STATE Compute $\hat{\textbf{x}}_i(0)$ using \eqref{eq:x_estimate_initial}.
\ENDFOR
\STATE \textbf{Iterative position estimation}:
\REPEAT
\FORALL {$i\in \mathcal{V}_x$ [in parallel] }
\STATE Update sensor $i$'s position $\hat{\textbf{x}}_i(\tau)$ as \eqref{eq:x_estimate_distribute}.
\IF {$\parallel R_i^2(\tau+1)-R_i^2(\tau)\parallel\leq\epsilon$}
\STATE Mark sensor $i$ as `localized'.
\ENDIF
\ENDFOR
\STATE $\tau=\tau+1$
\UNTIL {All sensors are `localized'}
\end{algorithmic} }}
\end{algorithm}
\vspace{-25pt}
\subsection{Convergence Analysis}
A key convergence property of Algorithm 1 is provided by the following theorem.
\begin{theorem}\label{theorem:convergence}
Let $\{\hat{\textbf{x}}_i(\tau)\}_{\tau=0}^{\infty}, i\in \mathcal{V}_x$ be the sequence of sensor $i$'s position estimates generated by Algorithm 1 and let the corresponding upper bounds of the squared position estimation errors be $\{R_i^2(\tau)\}_{\tau=0}^{\infty}$. Then for every $\tau$, we have
\begin{align}\label{eq:boundConvergence}
R_i^2(\tau+1)\leq R_i^2(\tau)
\end{align}
and
\begin{equation}\label{eq:R_converge}
\mathop {\lim }\limits_{\tau  \to \infty }(R_i^2(\tau+1)-R_i^2(\tau))=0
\end{equation}
and
\begin{equation}\label{eq:x_converge}
\mathop {\lim }\limits_{\tau  \to \infty }\parallel\hat {\textbf{x}}_i(\tau+1)-\hat {\textbf{x}}_i(\tau)\parallel=0
\end{equation}
\end{theorem}
\begin{proof}

For sensor $i$, at the $(\tau+1)$-th iteration, we have
\begin{align}\label{eq:bound_t+1}
R_i^2(\tau+1)=\mathop {\min }\limits_{{\hat {\textbf{x}}_i}} \mathop {\max }\limits_{({\textbf{y}_i},\Delta_i)\in\mathcal{Q}_i(\tau)} \{\Delta_i-2\textbf{y}_i\hat{\textbf{x}}_i^T+\hat{\textbf{x}}_i\hat{\textbf{x}}_i^T\}
\end{align}
where
\begin{align}
\mathcal{Q}_i(\tau)=\{({\textbf{y}_i},\Delta_i):  &\Delta_i-2\textbf{y}_i\hat{\textbf{x}}_i(\tau)^T+\hat{\textbf{x}}_i(\tau)\hat{\textbf{x}}_i(\tau)^T\leq R_i(\tau)^2 \nonumber\\
 &\underline{d_{ij}}^2 \le \textit{g}_{ij}(\tau)\le\overline{d_{ij}}^2 ,\forall j \in {{\cal N}_i} \nonumber\\
&\Delta_i\geq\textbf{y}_i\textbf{y}_i^T\} \nonumber
\end{align}
By changing the order of optimization and solving the minimization part, in a similar way with \eqref{eq:singleNode_replace}, the maximization part becomes
\begin{align}
R_i^2(\tau+1)=\mathop {\max }\limits_{({\textbf{y}_i},\Delta_i)\in\mathcal{Q}_{i}(\tau)}(\Delta_i-\textbf{y}_i\textbf{y}_i^T)
\end{align}
The optimal value of $\textbf{y}_i$ is $\hat{\textbf{x}}_i(\tau+1)$ in \eqref{eq:x_estimate_distribute}. Then, we have
\begin{align}
R_i^2(\tau+1)=\mathop {\max }\limits_{\Delta_i\in\mathcal{S}_{i}(\tau)}(\Delta_i-\hat{\textbf{x}}_i(\tau+1)\hat{\textbf{x}}_i(\tau+1)^T)
\end{align}
where
\begin{align*}
\mathcal{S}_{i}(\tau)=\{\Delta_i:~  &\Delta_i-2\hat{\textbf{x}}_i(\tau+1)\hat{\textbf{x}}_i(\tau)^T+\hat{\textbf{x}}_i(\tau)\hat{\textbf{x}}_i(\tau)^T\leq R_i(\tau)^2 \nonumber\\
 &\underline{d_{ij}}^2 \le \Delta_i-2\hat{\textbf{x}}_i(\tau+1)\hat{\textbf{x}}_j(\tau)^T+\hat{\textbf{x}}_j(\tau)\hat{\textbf{x}}_j(\tau)^T\le\overline{d_{ij}}^2 ,\forall j \in {{\cal N}_i} \nonumber\\
&\Delta_i\geq\hat{\textbf{x}}_i(\tau+1)\hat{\textbf{x}}_i(\tau+1)^T\}
\end{align*}
Let $u_i(\tau)=R_i(\tau)^2+2\hat{\textbf{x}}_i(\tau+1)\hat{\textbf{x}}_i(\tau)^T-\hat{\textbf{x}}_i(\tau)\hat{\textbf{x}}_i(\tau)^T$ and $u_{ij}(\tau)=\overline{d_{ij}}^2+2\hat{\textbf{x}}_i(\tau+1)\hat{\textbf{x}}_j(\tau)^T-\hat{\textbf{x}}_j(\tau)\hat{\textbf{x}}_j(\tau)^T$; then the value of $R_i^2(\tau+1)$ is
\begin{align}\label{eq:R_i_plus}
R_i^2(\tau+1)&=\min\{u_i(\tau),u_{ij}(\tau),\forall j\in\mathcal{N}_i\}-\hat{\textbf{x}}_i(\tau+1)\hat{\textbf{x}}_i(\tau+1)^T\nonumber\\
&\leq u_i(\tau)-\hat{\textbf{x}}_i(\tau+1)\hat{\textbf{x}}_i(\tau+1)^T \nonumber\\
&= R_i^2(\tau)-\parallel\hat{\textbf{x}}_i(\tau+1)-\hat{\textbf{x}}_i(\tau)\parallel^2
\end{align}
Straightforwardly, we can easily obtain \eqref{eq:boundConvergence}. By recursion, we have
\begin{align}
R_i^2(\tau+1)\leq R_i^2(\tau)-\parallel\hat{\textbf{x}}_i(\tau+1)-\hat{\textbf{x}}_i(\tau)\parallel^2\leq \cdots \leq R_i^2(0)-\sum\limits_{l = 0}^{\tau}  {\parallel\hat{\textbf{x}}_i(l+1)-\hat{\textbf{x}}_i(l)\parallel^2}
\end{align}
Since $R_i^2(\tau+1)\geq0$, we have $\sum\limits_{l = 1}^{\infty}  {\parallel\hat{\textbf{x}}_i(l+1)-\hat{\textbf{x}}_i(l)\parallel^2}\leq R_i^2(0)$, which means an infinite sum of non-negative values is bounded. Therefore, \eqref{eq:x_converge} must hold.

From \eqref{eq:R_i_plus}, $\forall j \in {\mathcal{N}_i}$, if $\mathop {\lim }\limits_{\tau  \to \infty } {u_{ij}}(\tau ) \le \mathop {\lim }\limits_{\tau  \to \infty } {u_i}(\tau )$,
\begin{align}\label{eq:limitR1}
\mathop {\lim }\limits_{\tau  \to \infty } R_i^2(\tau  + 1) &=\mathop {\lim }\limits_{\tau  \to \infty }\mathop {\min}\limits_{j \in {\mathcal{N}_i}} {u_{ij}}(\tau  + 1) - {{\hat {\textbf{x}}}_i}(\tau  + 1){{\hat {\textbf{x}}}_i}{{(\tau  + 1)}^T} \nonumber\\
&= \mathop {\lim }\limits_{\tau  \to \infty }\mathop {\min}\limits_{j \in {\mathcal{N}_i}} \left( {\overline {d _{ij}}^2 - {{\parallel {{{\hat {\textbf{x}}}_i}(\tau  + 1) - {{\hat {\textbf{x}}}_j}(\tau )} \parallel}^2}} \right)
\end{align}
If $\mathop {\lim }\limits_{\tau  \to \infty } {u_{ij}}(\tau ) > \mathop {\lim }\limits_{\tau  \to \infty } {u_i}(\tau )$,
\begin{align}\label{eq:limitR2}
\mathop {\lim }\limits_{\tau  \to \infty } R_i^2(\tau  + 1) &=\mathop {\lim }\limits_{\tau  \to \infty }{u_{i}}(\tau) - {{\hat {\textbf{x}}}_i}(\tau  + 1){{\hat {\textbf{x}}}_i}{{(\tau  + 1)}^T} \nonumber\\
&= \mathop {\lim }\limits_{\tau  \to \infty }R_i^2(\tau ) - {{\parallel {{{\hat {\textbf{x}}}_i}(\tau  + 1) - {{\hat {\textbf{x}}}_i}(\tau )} \parallel}^2}
\end{align}
Since $\mathop {\lim }\limits_{\tau  \to \infty }{\textbf{x}}_i(\tau +1)=\mathop {\lim }\limits_{\tau  \to \infty }{\textbf{x}}_i(\tau), \forall i \in \mathcal{V}_x$, from \eqref{eq:limitR1} and \eqref{eq:limitR2}, we always have $\mathop {\lim }\limits_{\tau  \to \infty }R_i^2(\tau+1 )=\mathop {\lim }\limits_{\tau  \to \infty }R_i^2(\tau )$. Thus \eqref{eq:R_converge} holds.
\end{proof}
\begin{remark}
Theorem 1 shows that the upper bounds of the position estimation errors generated by Algorithm 1, i.e., $\{R_i^2(\tau)\}_{\tau=0}^{\infty},~\forall i\in \mathcal{V}_x$, are non-increasing positive sequences. When $\tau \to \infty$, $R_i^2(\tau)$ will converge to a fixed value, not necessarily 0, and $\textbf{x}_i(\tau)$ satisfies $\mathop {\lim }\limits_{\tau \to\infty }\parallel\textbf{x}_i(\tau)-\textbf{x}_i\parallel^2\leq \mathop {\lim }\limits_{\tau  \to \infty }R_i^2(\tau)$. The converged value of $R_i^2(\tau)$ is determined by the network configuration and measurement errors.
\end{remark}
\subsection{Computational Complexity and Communication Cost}

In the initial estimation, a sensor only needs to broadcast the necessary information once. The information that sensor $i$ broadcasts to its neighbors is $\{\underline{d_{ik}}, \overline{d_{ik}}, n_{ik}\}, \forall k\in \mathcal{N}_a$. Therefore, the communication cost for sensor $i$ is $3m\mid\mathcal{N}_{ij}\mid$. The computational cost comes from the SDP problem in \eqref{eq:SDP_initial}, in which there are $2m+1$ scalar variables to be optimized. The computational complexity of initial estimation is $\mathcal{{O}}(m^3)$.

In the iterative estimation, at each iteration, for sensor $i$, the computational cost comes from the SDP problem in \eqref{eq:SDP_distributed}, in which there are $2|\mathcal{N}_i|+2$ scalar variables to be optimized, and the number of scalar inequality constraints is $2|\mathcal{N}_i|+2$, the number of LMI constraints is $1$, the size of the LMI is at most $3\times3$. Therefore, the computational complexity of \eqref{eq:SDP_distributed} is $\mathcal{{O}}(|\mathcal{N}_i|^3)$ \cite{boyd2004convex}. The communication cost at each iteration for sensor $i$ comes from the information exchange. Senor $i$ needs to send its position estimate at each iteration to its neighbor nodes. For 2-D localization, the communication cost for sensor $i$ is proportional to $2|\mathcal{N}_i|$.

\section{Simulation Results} \label{sec:simulations}
In this section, we conduct extensive simulations to illustrate the performance of the relaxed estimation in Section \ref{sec:minmax_algorithm}, denoted by MinMax-SDP, and the distributed algorithm in Section \ref{sec:distributed alg}, denoted by Dis-MinMax.

\subsection{Performance of Centralized Algorithm}
To investigate the performance of our proposed centralized algorithm, we first compare our algorithm with three centralized algorithms when the measurement errors follow two well-known distributions, i.e., uniform distribution and Gaussian distribution \cite{MLE2014TSP}; second, we consider a mixture error model with Gaussian distributed errors as inliers and uniformly distributed errors as outliers; third, we test the performance of our proposed algorithm using the experimental data provided by Patwari \emph{et al}. in \cite{patwari2003TSP}.

Our simulations are conducted in a unit square area. In this area, $50$ sensors with unknown positions are randomly deployed using a uniform distribution, along with four anchors whose positions are known as $(-0.3,-0.3)$, $(0.3,-0.3)$, $(-0.3,0.3)$ and $(0.3,0.3)$. To guarantee the unique localizability of this network, we set the sensing range of each node as $R=0.5$.

\subsubsection{Measurement error follows uniform or Gaussian distribution}
In this part we compare MinMax-SDP with three centralized algorithms: \emph{a}) SDP in \cite{biswas2006semidefinite}, which minimizes the norm of the error between the squared range measurements and the squared range estimates; \emph{b}) ESDP in \cite{ESDP}, which involves a further relaxation of SDP, and is of less computational complexity than SDP; \emph{c}) E-ML in \cite{MLE2014TSP}, which is an MLE with ESDP relaxation.

In order to show the robustness to measurement error distributions, we consider two scenarios: 1) the measurement errors are uniformly distributed ; 2) the measurement errors are Gaussian distributed. That is, the actual measurement errors are determined in accord with these statistics, but our algorithm runs under the assumption that the errors are limited by a fixed known bound. The bound will be related to the statistics of the actual measurement errors.

Fig.~\ref{fig:locMapGama0.02_uniform} shows the localization results when the measurement errors are uniformly distributed in $[-\gamma,\gamma]$ and $\gamma=0.02$. We can see that the localization errors of the four algorithms are all very small. The root mean squared error (RMSE) of MinMax-SDP is slightly higher than SDP, but smaller than ESDP and E-ML. However, when $\gamma$ increases to $0.1$, from Fig.~\ref{fig:locMapGama0.1_uniform}, all of SDP, ESDP, and E-ML collapse, because the projection of high-dimension solution onto a lower dimension results in points getting `crowded' together when measurement errors become large\cite{biswas2006semidefinite}. MinMax-SDP can still work with increased localization error. The reason behind this phenomenon is that when the measurement errors are large, our proposed algorithm can guarantee that the estimation error is bounded as shown in \eqref{eq:error_bound}. However, SDP, ESDP and E-ML are known to collapse when the measurement errors become large \cite{anderson2010formal}. When we set a larger measurement error bound, the worst-case measurement error would be larger and the probability for getting larger measurement errors increases. Therefore, our proposed algorithm outperforms SDP, ESDP and E-ML when the measurement error bound becomes large.

To compare the statistical performance of these algorithms, we conduct 50 Monte Carlo trials. Fig.~\ref{fig:RMSE_uniform} shows the RMSEs of these four algorithms under different measurement error bounds, from which, we can see, MinMax-SDP performs the best, and the advantage is much more obvious when the measurement error bound becomes larger. When the measurement errors follow Gaussian distribution, with zero-mean and variance as $\sigma^2$, the measurement error bound utilized in MinMax-SDP is taken as $\gamma=3\sigma$. The only information about the measurement errors for MinMax-SDP, SDP and ESDP is the bound, while E-ML has full knowledge of the measurement error distribution. From the RMSEs shown in Fig.~\ref{fig:RMSE_Gauss}, we can see that the performance of MinMax-SDP is comparable with that of E-ML, and much better than that of SDP and ESDP.

 \begin{figure}[!ht]
  \centering
\subfigure[MinMax-SDP]{
 \begin{minipage}[t]{0.35\linewidth}
 \includegraphics[width=\textwidth]{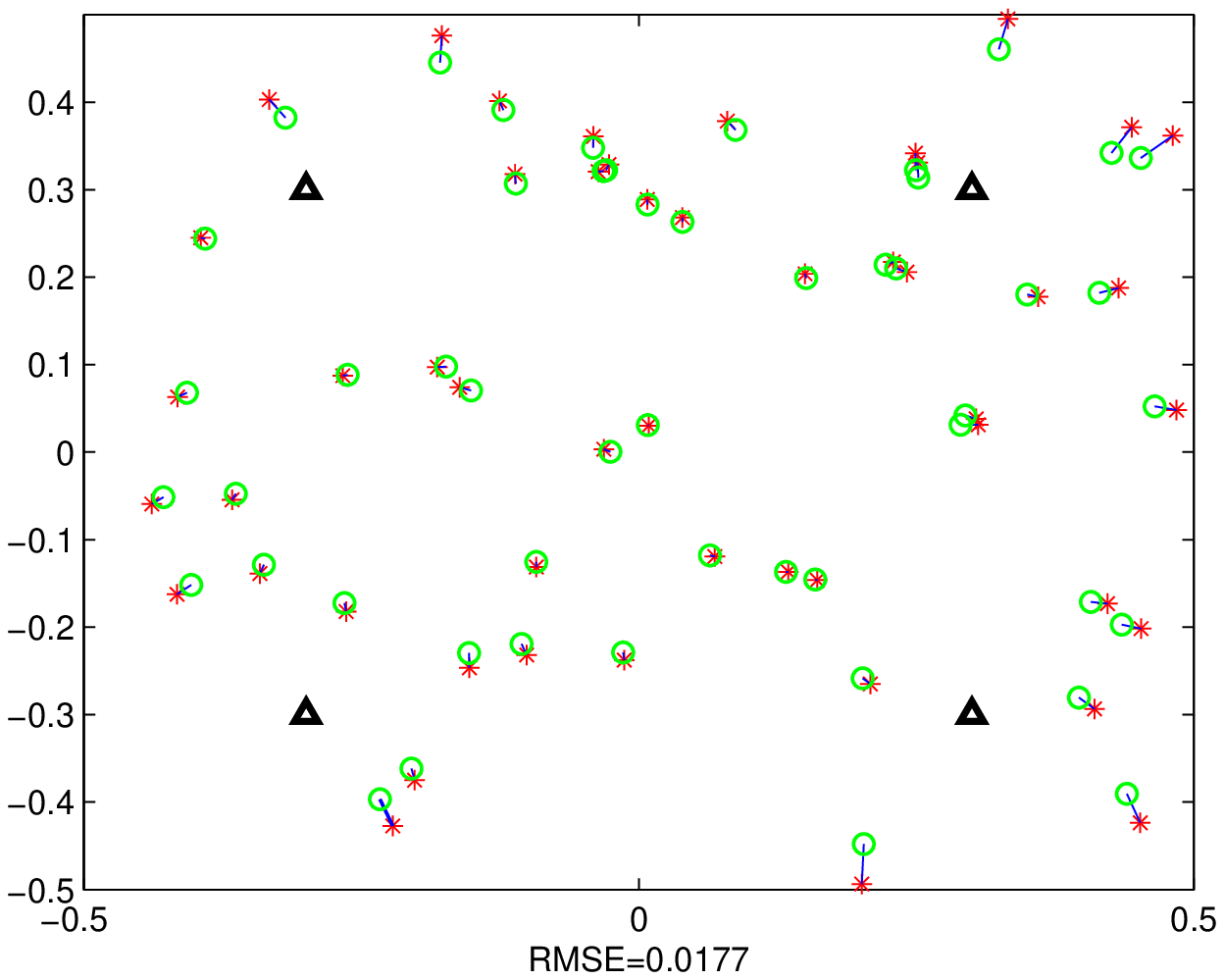}
 \label{subfig:minmax_uniform_Gama0.02}
 \vspace{-15pt}
 \end{minipage}}
\subfigure[SDP]{
 \begin{minipage}[t]{0.35\linewidth}
 \includegraphics[width=\textwidth]{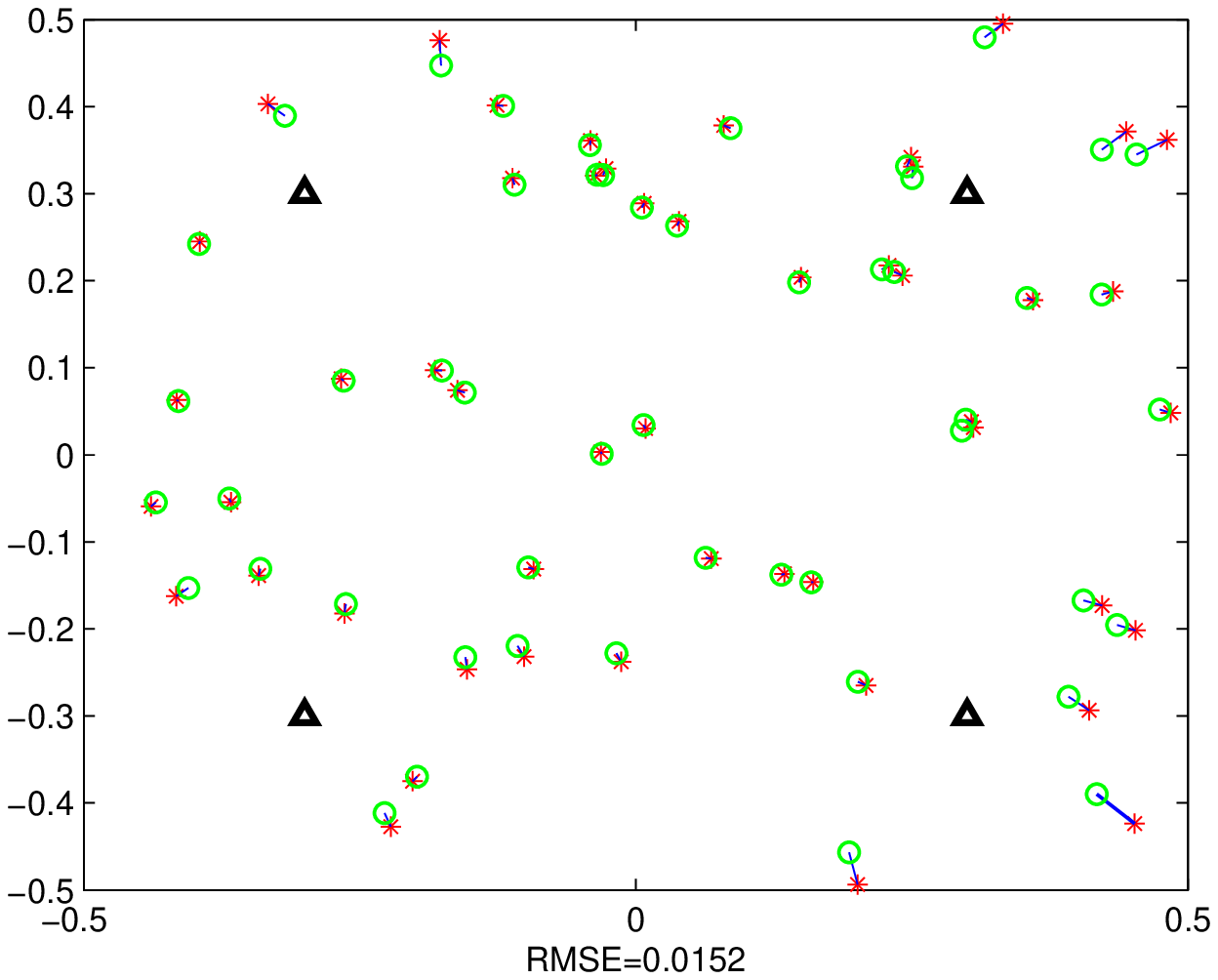}
  \label{subfig:SDP_uniform_Gama0.02}
  \vspace{-15pt}
 \end{minipage}}
 \subfigure[ESDP]{
 \begin{minipage}[t]{0.35\linewidth}
 \includegraphics[width=\textwidth]{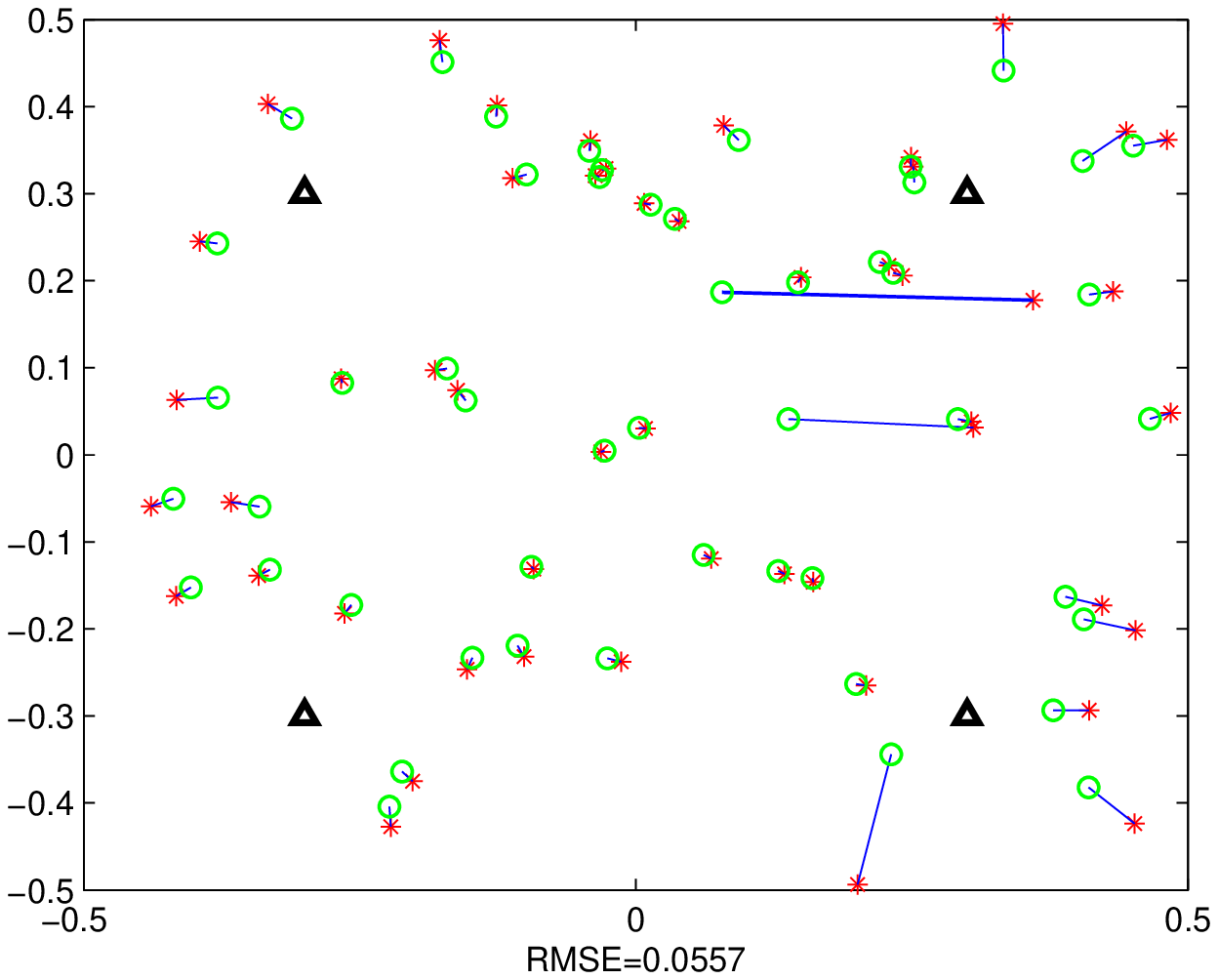}
  \label{subfig:ESDP_uniform_Gama0.02}
  \vspace{-15pt}
 \end{minipage}}
 \subfigure[E-ML]{
 \begin{minipage}[t]{0.35\linewidth}
 \includegraphics[width=\textwidth]{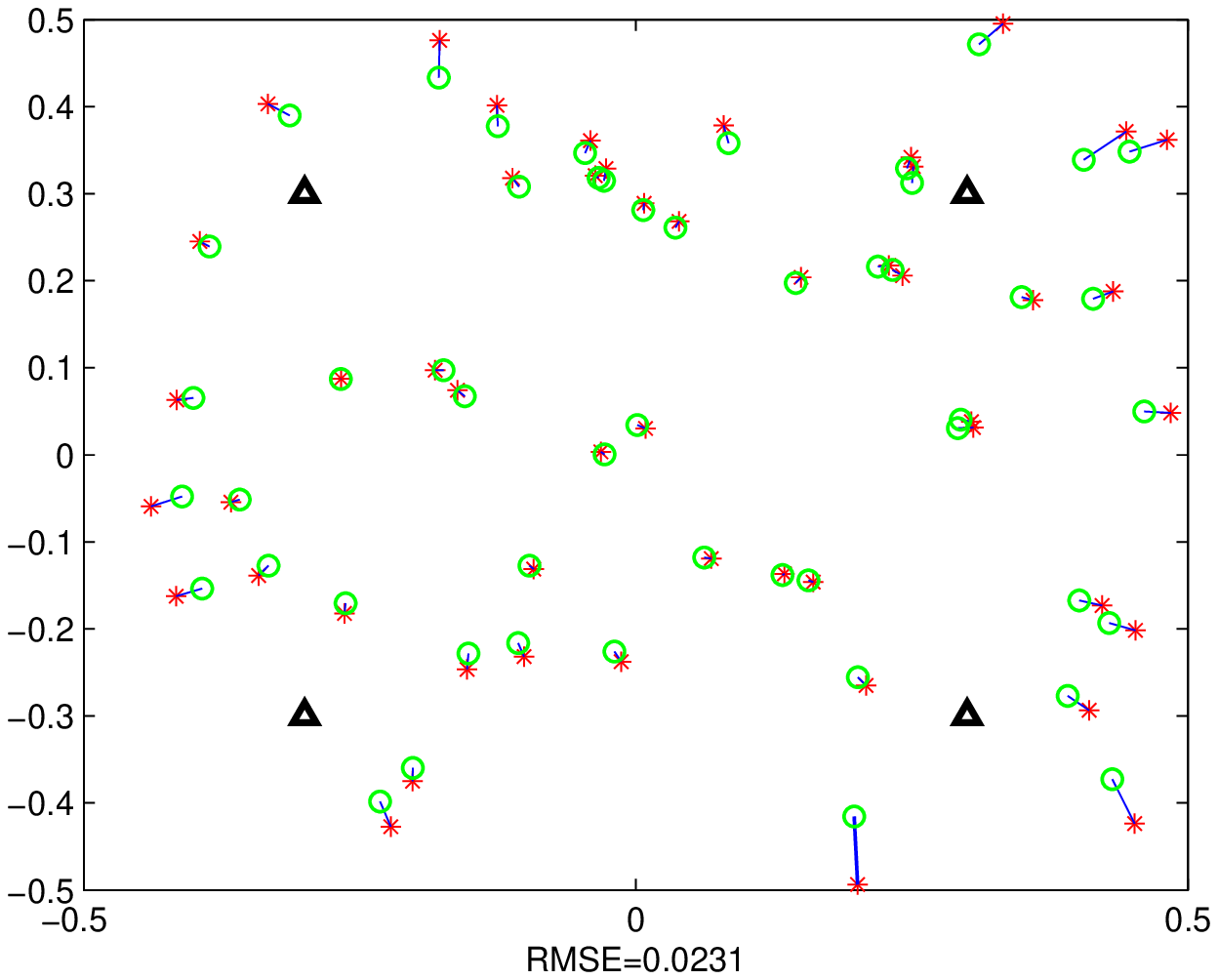}
  \label{subfig:EML_uniform_Gama0.02}
 \end{minipage}}
 \caption{Localization results under uniformly distributed measurement errors, where the triangles denote the anchors, the stars denote the true sensors' positions and circles denote the estimated sensors' positions. Measurement error bound $\gamma=0.02$.}
\label{fig:locMapGama0.02_uniform}
\vspace{-20pt}
 \end{figure}

  \begin{figure}[!ht]
  \centering
\subfigure[MinMax-SDP]{
 \begin{minipage}[t]{0.35\linewidth}
 \includegraphics[width=\textwidth]{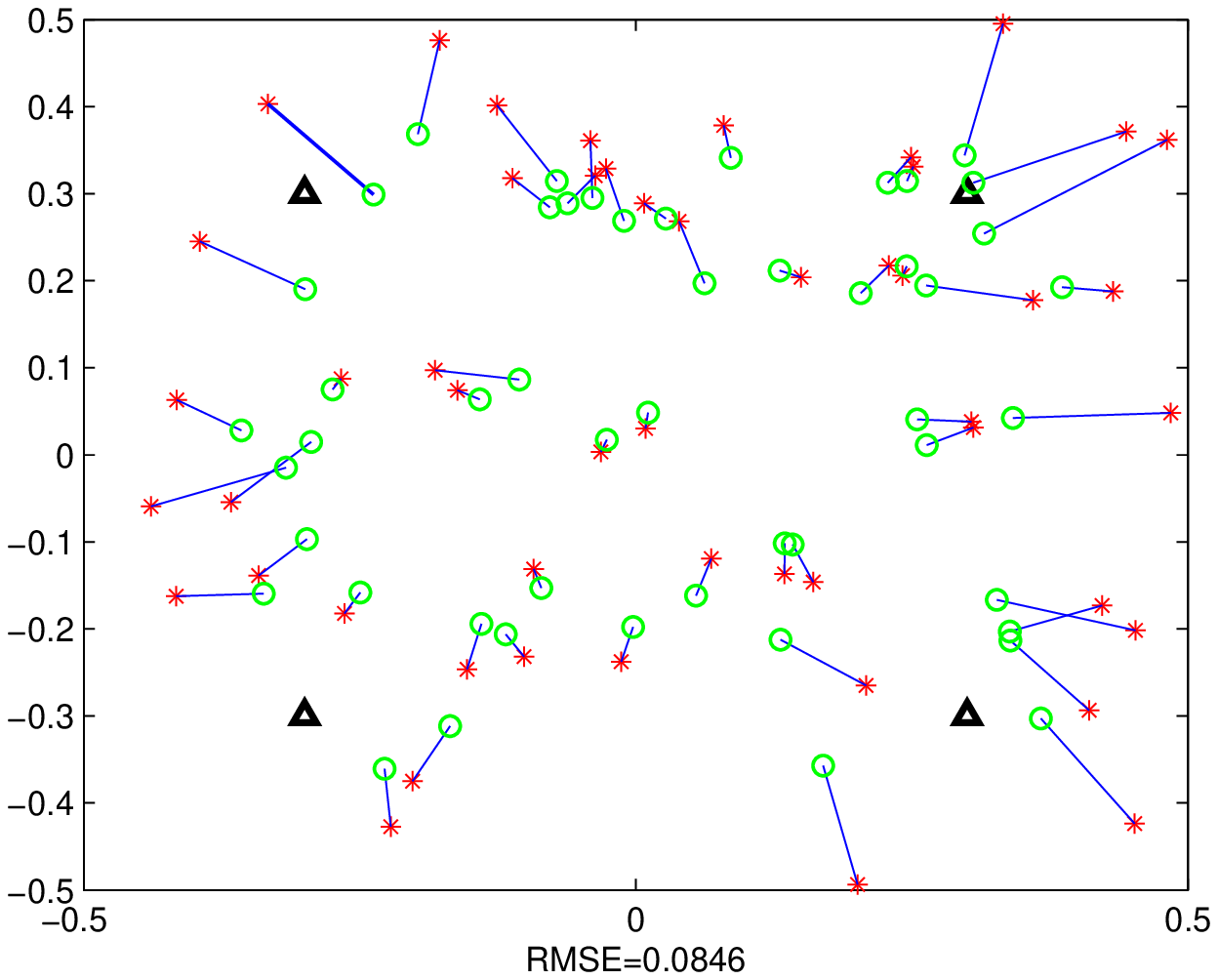}
 \label{subfig:minmax_uniform_Gama0.1}
 \vspace{-15pt}
 \end{minipage}}
\subfigure[SDP]{
 \begin{minipage}[t]{0.35\linewidth}
 \includegraphics[width=\textwidth]{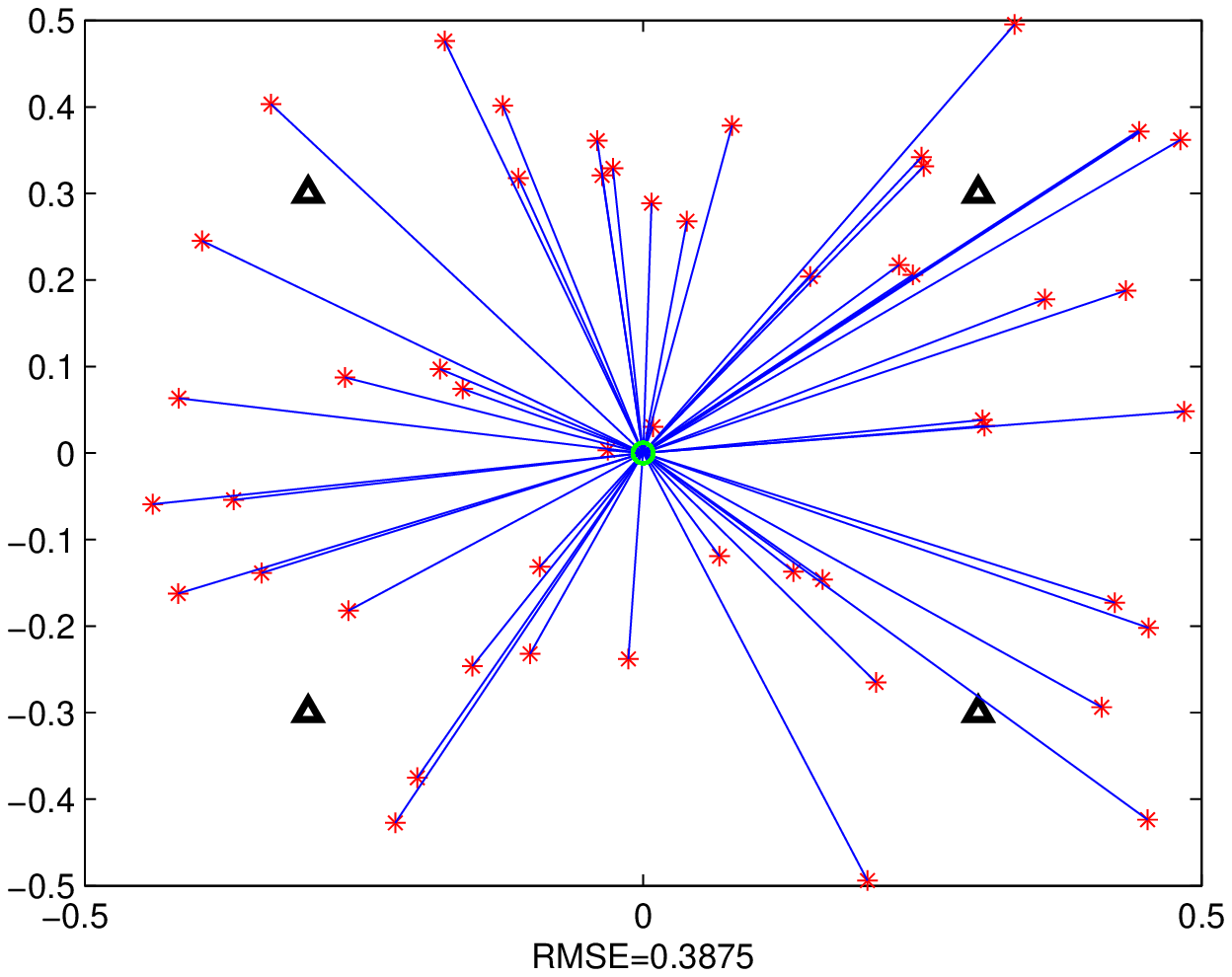}
  \label{subfig:SDP_uniform_Gama0.1}
  \vspace{-15pt}
 \end{minipage}}
 \subfigure[ESDP]{
 \begin{minipage}[t]{0.35\linewidth}
 \includegraphics[width=\textwidth]{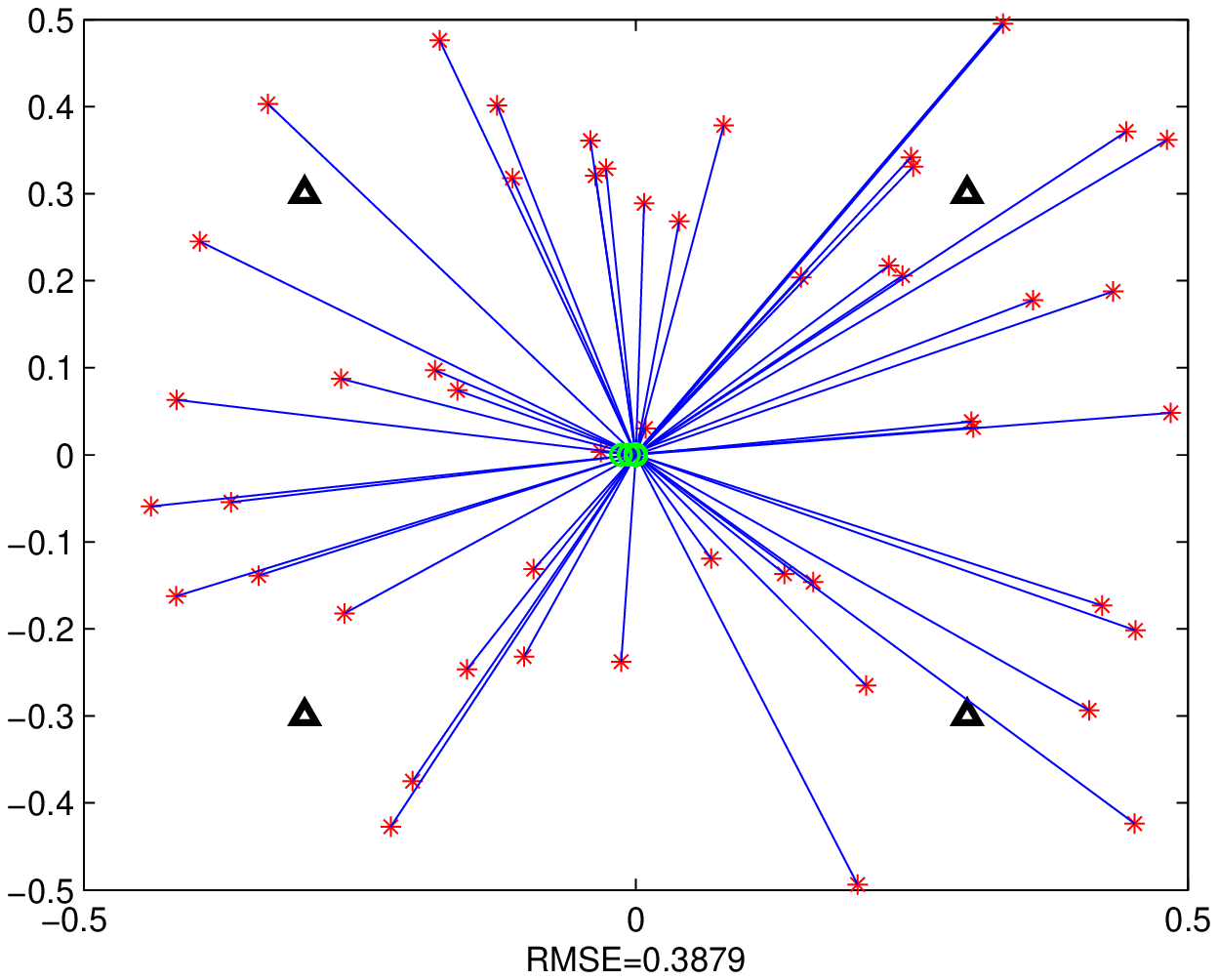}
  \label{subfig:ESDP_uniform_Gama0.1}
  \vspace{-15pt}
 \end{minipage}}
 \subfigure[E-ML]{
 \begin{minipage}[t]{0.35\linewidth}
 \includegraphics[width=\textwidth]{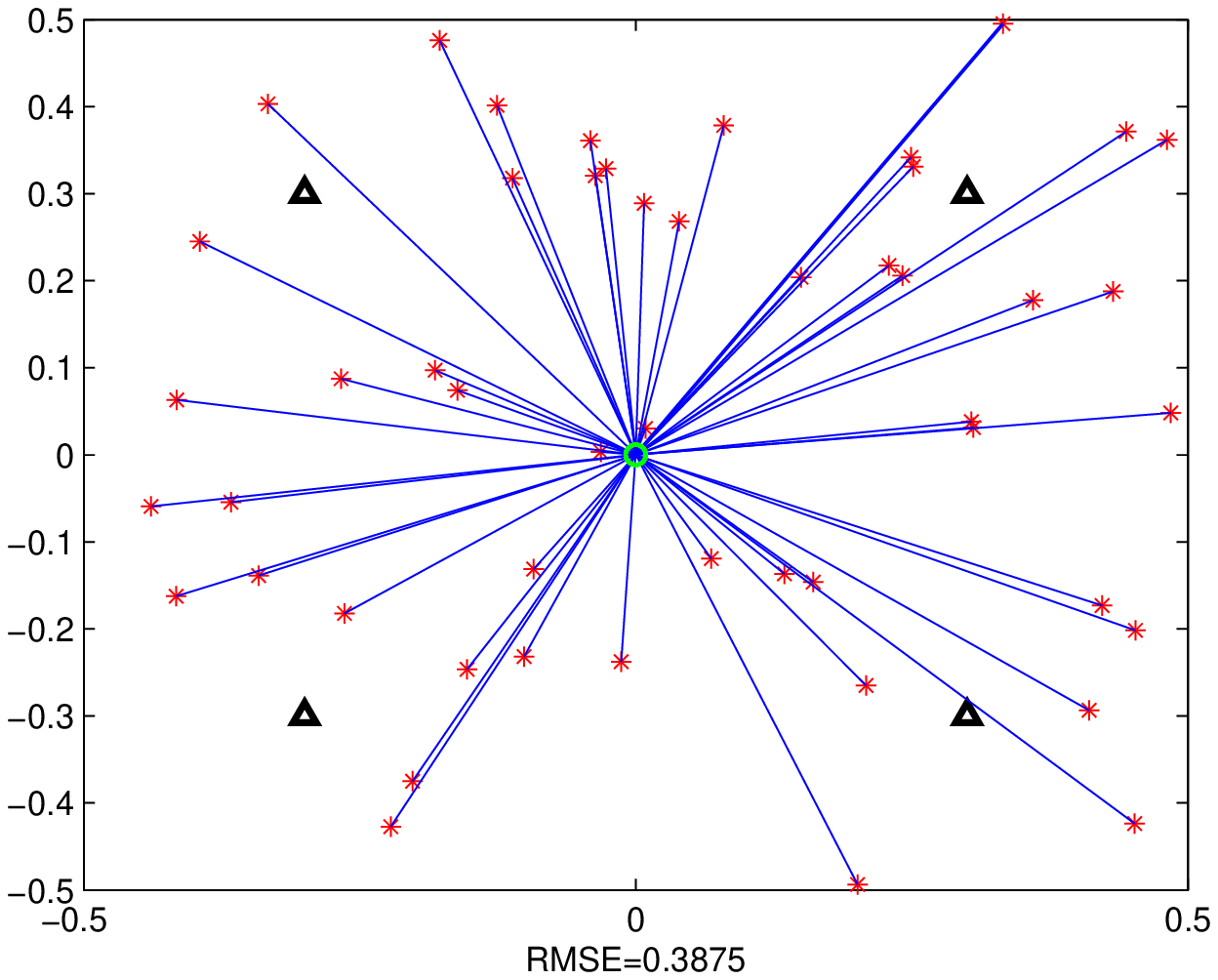}
  \label{subfig:EML_uniform_Gama0.1}
  \vspace{-15pt}
 \end{minipage}}
 \caption{Localization results under uniformly distributed measurement errors, where the triangles denote the anchors, the stars denote the true sensors' positions and circles denote the estimated sensors' positions. Measurement error bound $\gamma=0.1$.}
\label{fig:locMapGama0.1_uniform}
\vspace{-15pt}
 \end{figure}

 \begin{figure}[!ht]
  \centering
 \includegraphics[width=0.35\textwidth]{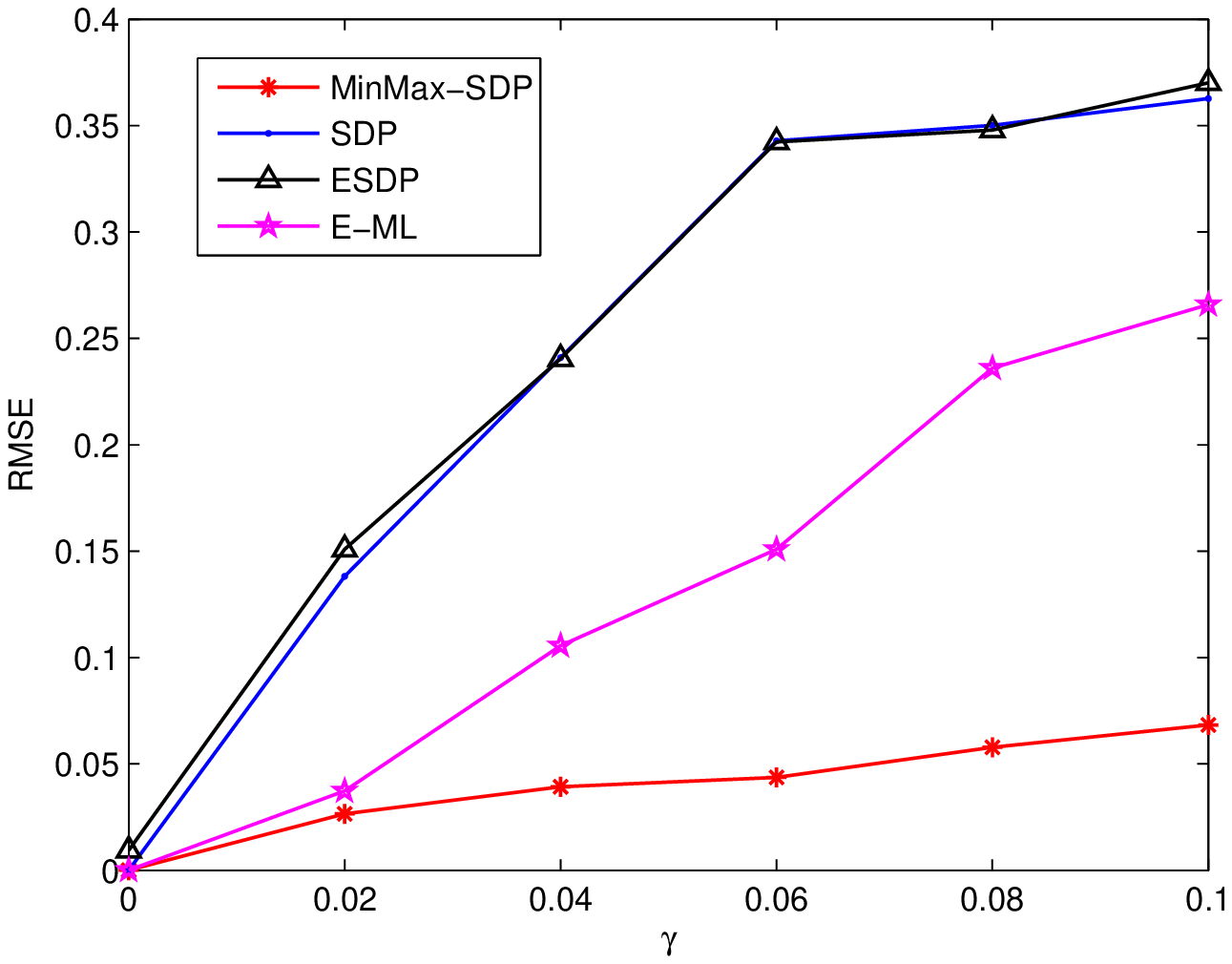}
 \vspace{-15pt}
 \caption{RMSEs under uniformly distributed measurement errors with different bounds in 50 trials.}
\label{fig:RMSE_uniform}
\vspace{-20pt}
 \end{figure}

By comparing Fig.~\ref{fig:RMSE_Gauss} with Fig.~\ref{fig:RMSE_uniform}, we find that, MinMax-SDP performs better than E-ML when the measurement errors are uniformly distributed, whereas it performs worse than E-ML when the measurement errors are Gaussian distributed.
At the first sight, this result appears to be counter-intuitive because E-ML, which has accurate knowledge of the measurement error distribution, should deliver better performance than the proposed scheme in both cases. We offer the following explanation for the observed result. When the measurement errors are uniformly distributed, all positions within the feasible region $\mathcal{C}$ occur with equal probability. An ML-based estimator cannot differentiate these positions and therefore may return any position within this feasible region as the position estimate. In comparison, our algorithm will only deliver the best position estimate that minimizes the worst-case estimation error, i.e. the one resembling the Chebyshev center. Therefore when the measurement errors are uniformly distributed, the proposed algorithm delivers better performance than an ML-based estimator. When the measurement errors are Gaussian distributed, all positions within the feasible region occur with different probabilities. In this situation, the accurate knowledge of the measurement error distribution, which forms the basis of E-ML estimator, can be exploited to deliver better performance than the proposed scheme, which does not rely on such knowledge.

 \begin{figure}[!ht]
  \centering
 \includegraphics[width=0.35\textwidth]{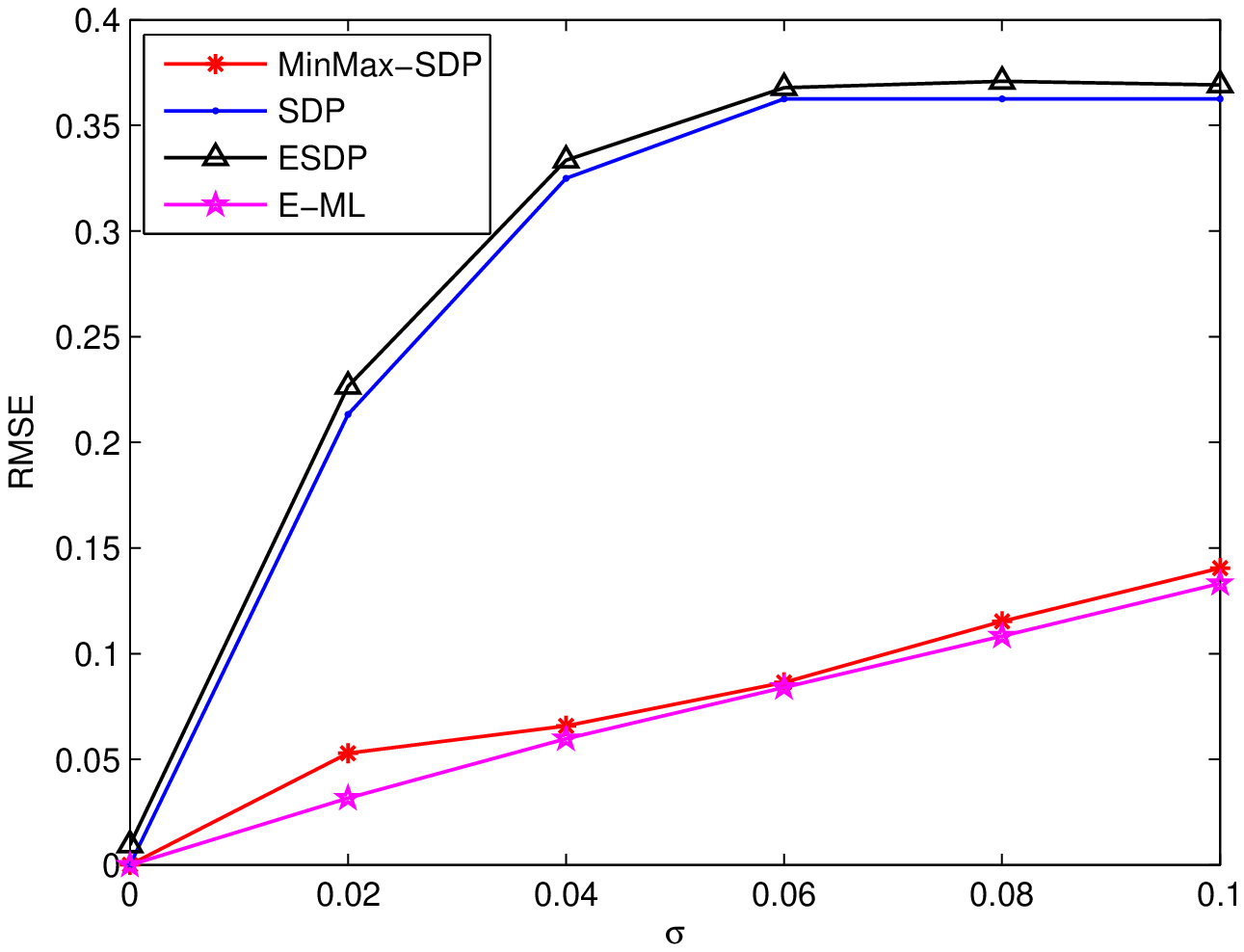}
 \vspace{-15pt}
 \caption{RMSEs under Gaussian distributed measurement errors with different standard deviations in 50 trials.}
\label{fig:RMSE_Gauss}
\vspace{-20pt}
 \end{figure}

 \subsubsection{Mixture measurement error model}
In practice, the measurement errors may not perfectly fit a statistical distribution, among which, usually there exist some outliers. In the simulation, the inliers follow a zero-mean Gaussian distribution with standard deviation $\sigma$, and the outliers are uniformly distributed in $[-3\sigma, 3\sigma]$.  Fig. \ref{fig:MixNoise} compares the performance of MinMax-SDP algorithm with E-ML. In MinMax-SDP, the error bound $\gamma=3\sigma$. In E-ML, the measurement errors are taken as Gaussian distributed with zero-mean and standard deviation $\sigma$. The ratio denotes the ratio between the number of uniformly distributed errors and the number of Gaussian distributed errors. We can see, when number of the outliers is small, i.e., $ratio=0.1$ in our figure, E-ML is better than MinMax-SDP when $\sigma<0.08$, but worse than MinMax-SDP when $\sigma>0.08$. When the ratio increases to $0.5$, we can see the performance of MinMax-SDP is slightly changed, but the performance of E-ML is much worse than that when $ratio=0.1$. We can conclude that MinMax-SDP is more robust to the changes of the measurement error distributions than E-ML.

 \begin{figure}[!ht]
 \centering
 \includegraphics[width=0.35\textwidth]{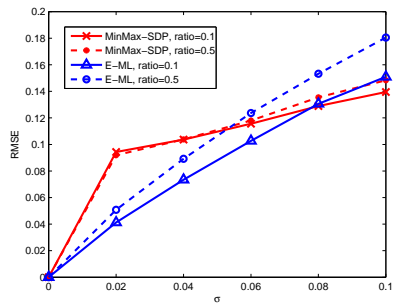}
 \vspace{-15pt}
 \caption{Localization results with mixture measurement error model, i.e., Gaussian distributed errors as inliers and uniformly distributed errors as outliers.}
\label{fig:MixNoise}
\vspace{-20pt}
 \end{figure}

 \subsubsection{Experimental evaluation}
In this part, we use experimental data provided in \cite{patwari2003TSP} to test the performance of MinMax-SDP. In the experiment, there are four devices placed near the corners in the area and 40 devices to be localized. The measurements utilized for localization are TOA measurements, of which, the measurement errors are zero-mean Gaussian distributed with standard deviation $\sigma_T=6.1ns$. Since the signal transmission speed $v$ is known, the range measurements can be easily obtained. In MinMax-SDP, the range measurement error bound is $3\sigma_Tv$. We assume that the sensing range of each device is $R=5m$, and build up the corresponding connectivity matrix. Only the measurements between the connected devices are used for localization. The localization results of MinMax-SDP and E-ML are shown in Fig. \ref{subfig:experimentMinMaxC} and Fig. \ref{subfig:experimentEML} respectively. Fig. \ref{subfig:experimentCDF} compares the CDF of the position estimation errors of MinMax-SDP and E-ML. We can see the performance of MinMax-SDP, which only uses the measurement error bound, is comparable with that of E-ML, which takes advantage of the statistical distribution of the measurement errors.

 \begin{figure}[!ht]
  \centering
\subfigure[MinMax-SDP]{
 \begin{minipage}[t]{0.34\linewidth}
 \includegraphics[width=\textwidth]{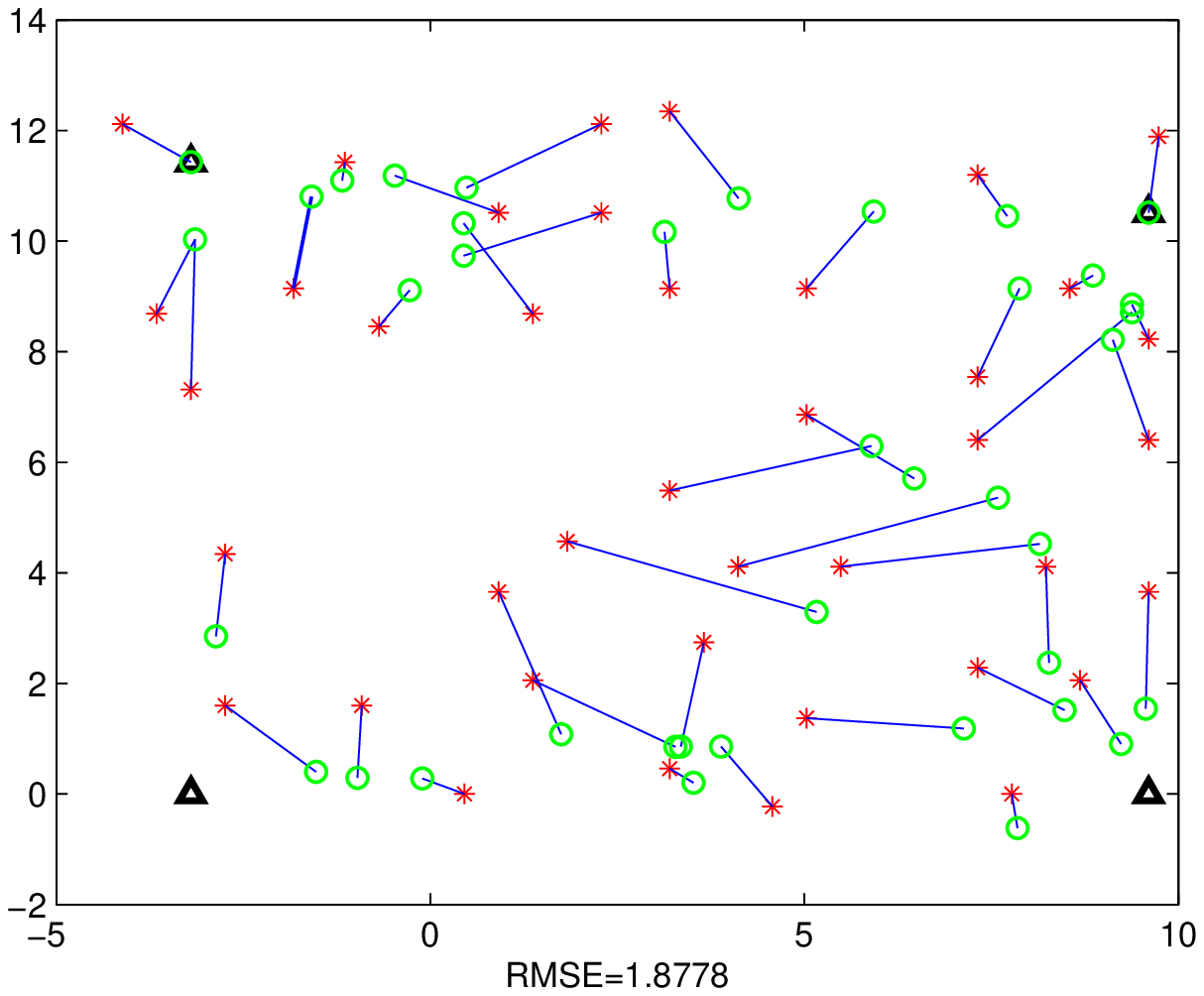}
 \label{subfig:experimentMinMaxC}
 \end{minipage}}
 \hspace{-20pt}
\subfigure[E-ML]{
 \begin{minipage}[t]{0.34\linewidth}
 \includegraphics[width=\textwidth]{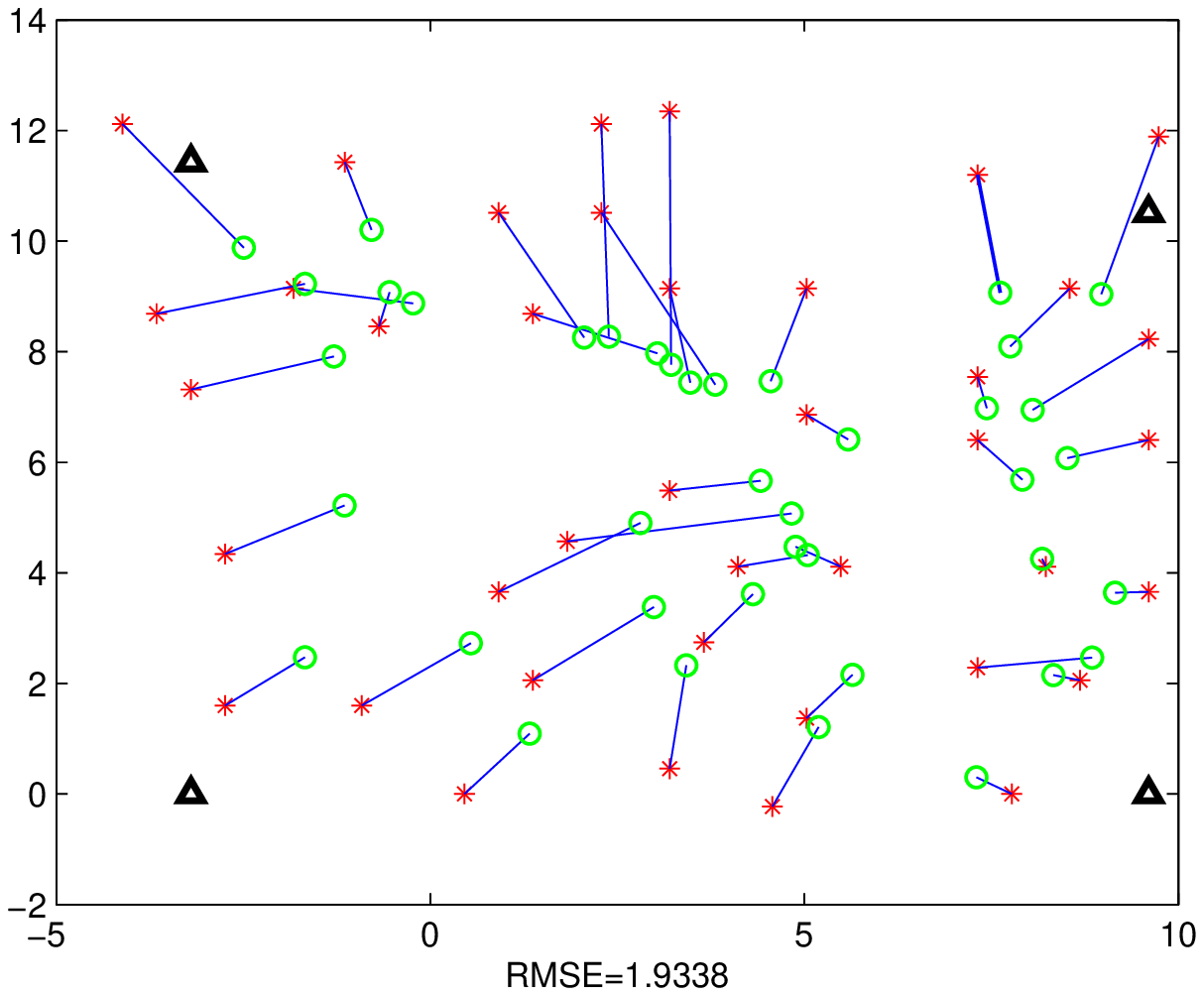}
  \label{subfig:experimentEML}
 \end{minipage}}
 \hspace{-20pt}
 \subfigure[CDF]{
 \begin{minipage}[t]{0.34\linewidth}
 \includegraphics[width=\textwidth]{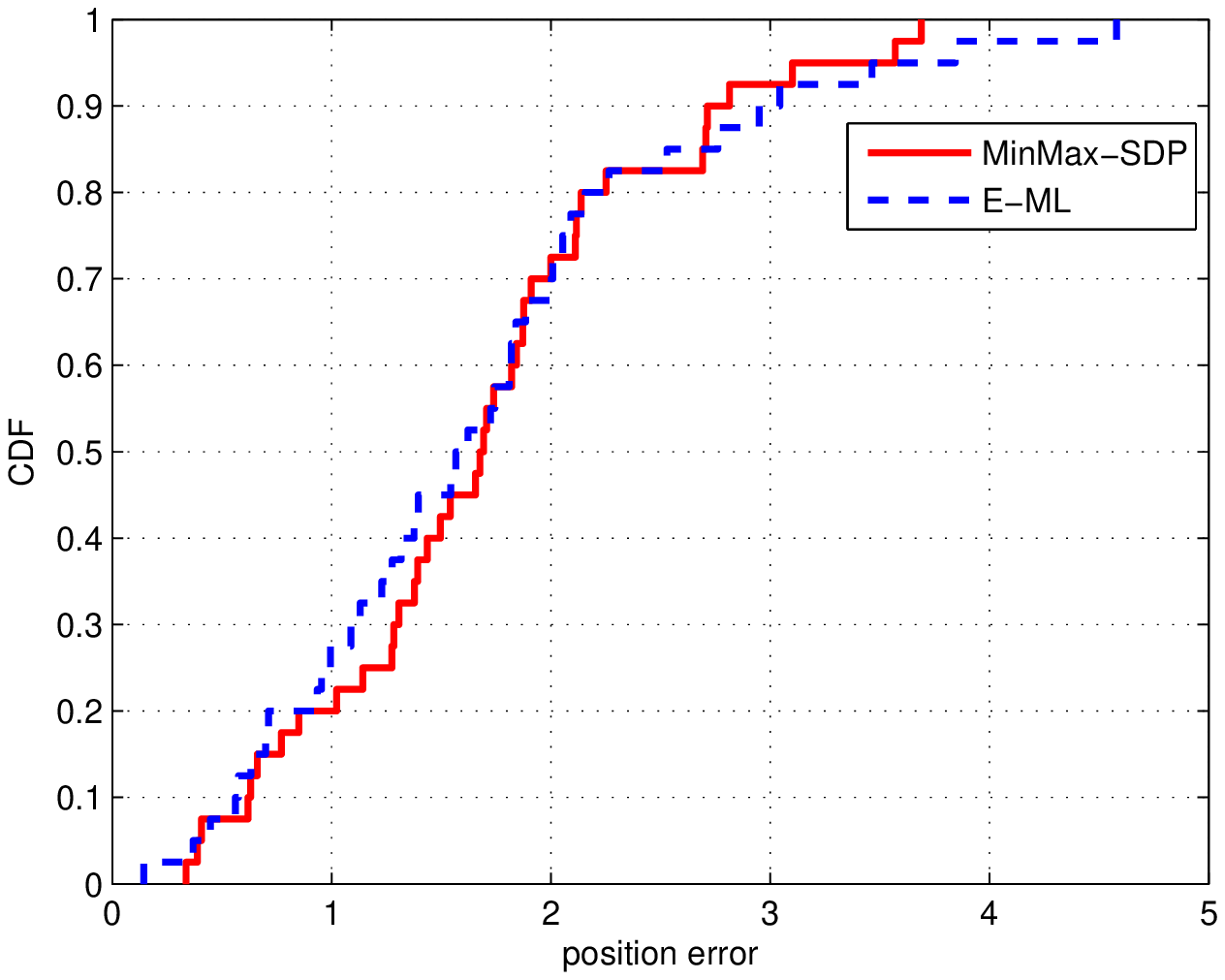}
  \label{subfig:experimentCDF}
 \end{minipage}}
 \caption{Experimental localization results under Gaussian distributed measurement errors, where the triangles denote the anchors, the stars denote the true positions of the sensors and the circles denote the estimated positions.}
\label{fig:Experimental results}
\vspace{-20pt}
 \end{figure}

\subsection{Performance of Distributed Algorithm}
 \begin{figure}[!ht]
  \centering
\subfigure[$n=50$, $R=0.5$]{
 \begin{minipage}[t]{0.33\linewidth}
 \includegraphics[width=\textwidth]{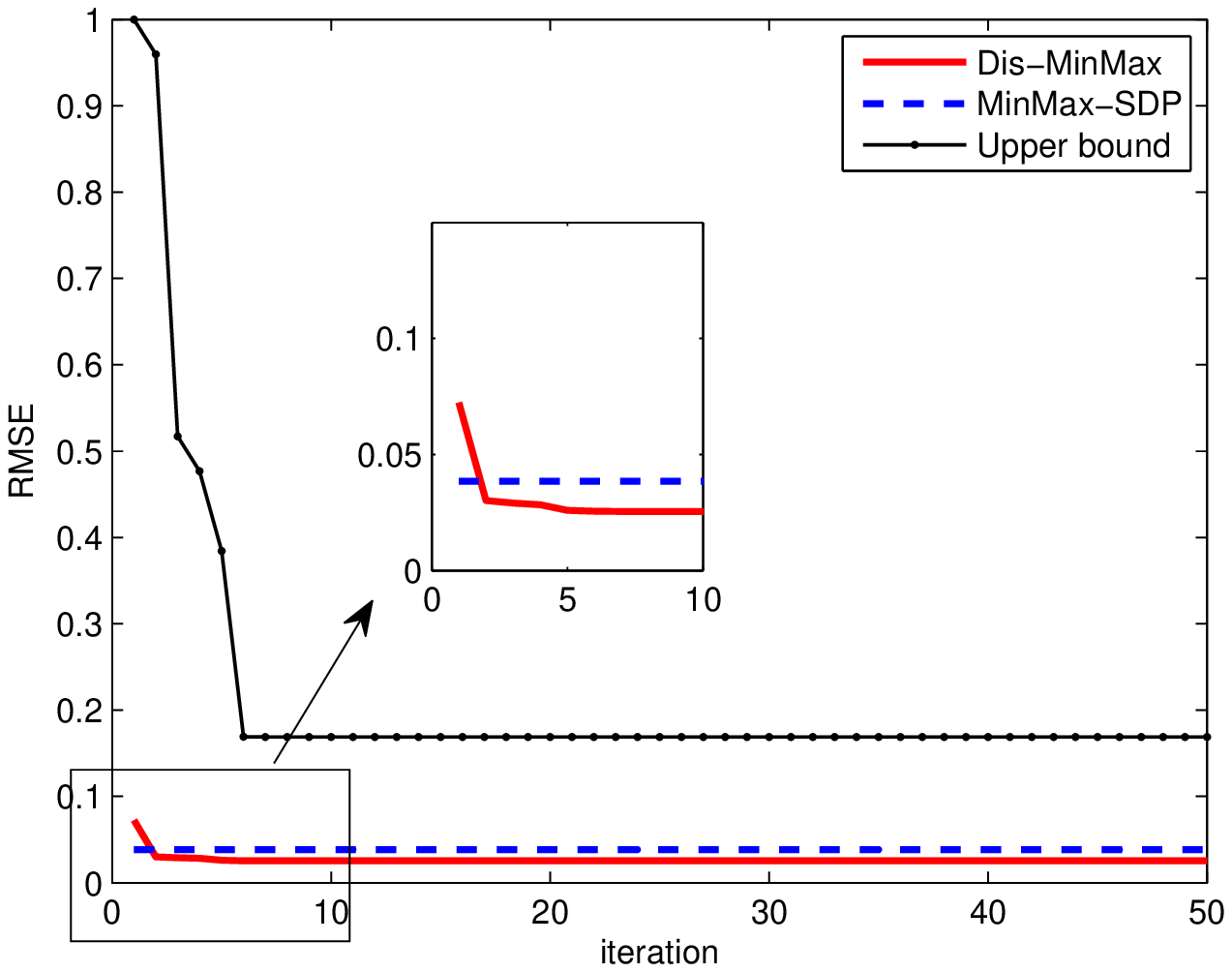}
 \label{subfig:distributedRMSE_N50}
 \end{minipage}}
  \hspace{-15pt}
\subfigure[$n=100$, $R=0.3$]{
 \begin{minipage}[t]{0.33\linewidth}
 \includegraphics[width=\textwidth]{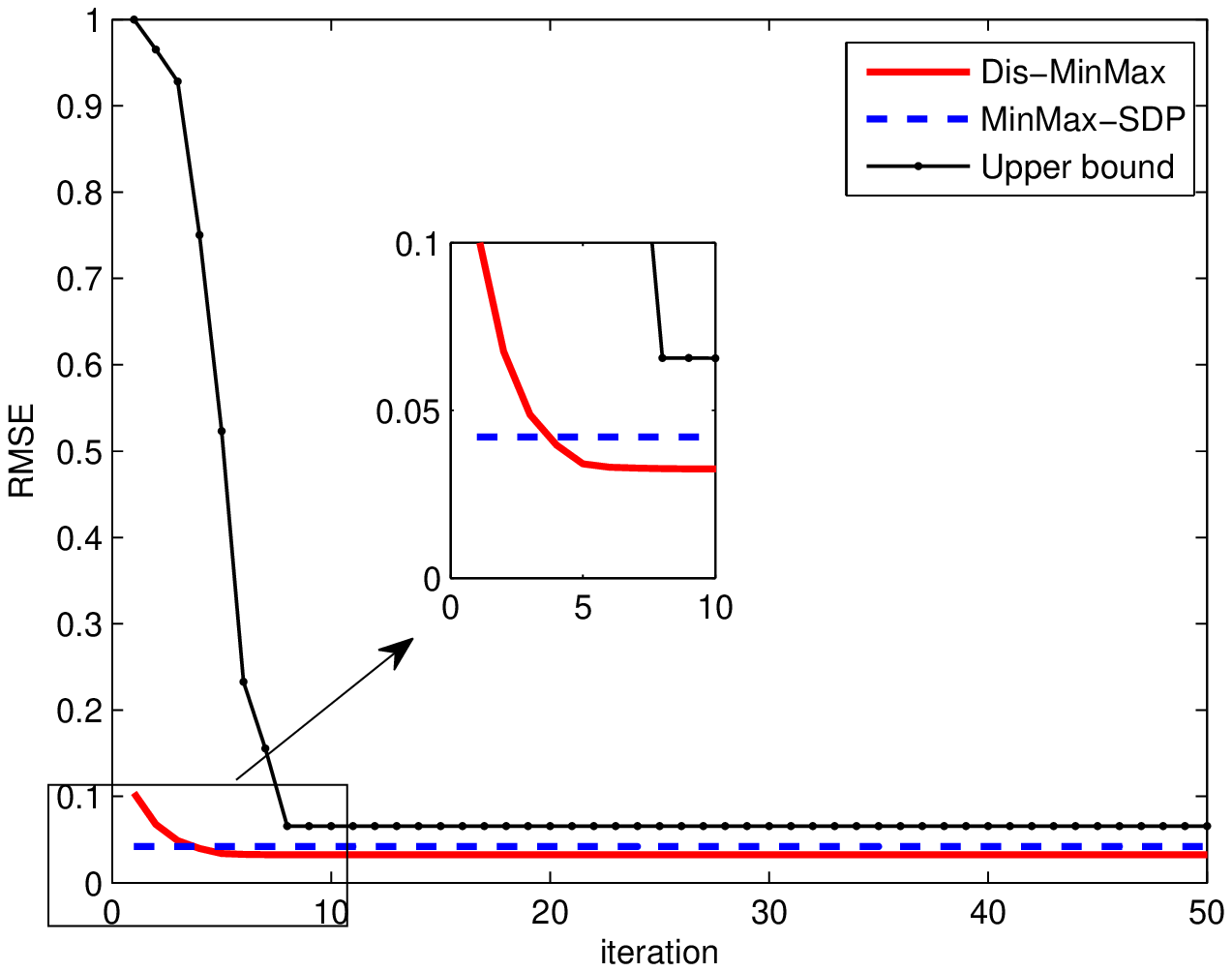}
  \label{subfig:distributedRMSE_N100}
 \end{minipage}}
   \hspace{-15pt}
 \subfigure[$n=200$, $R=0.2$]{
  \begin{minipage}[t]{0.33\linewidth}
 \includegraphics[width=\textwidth]{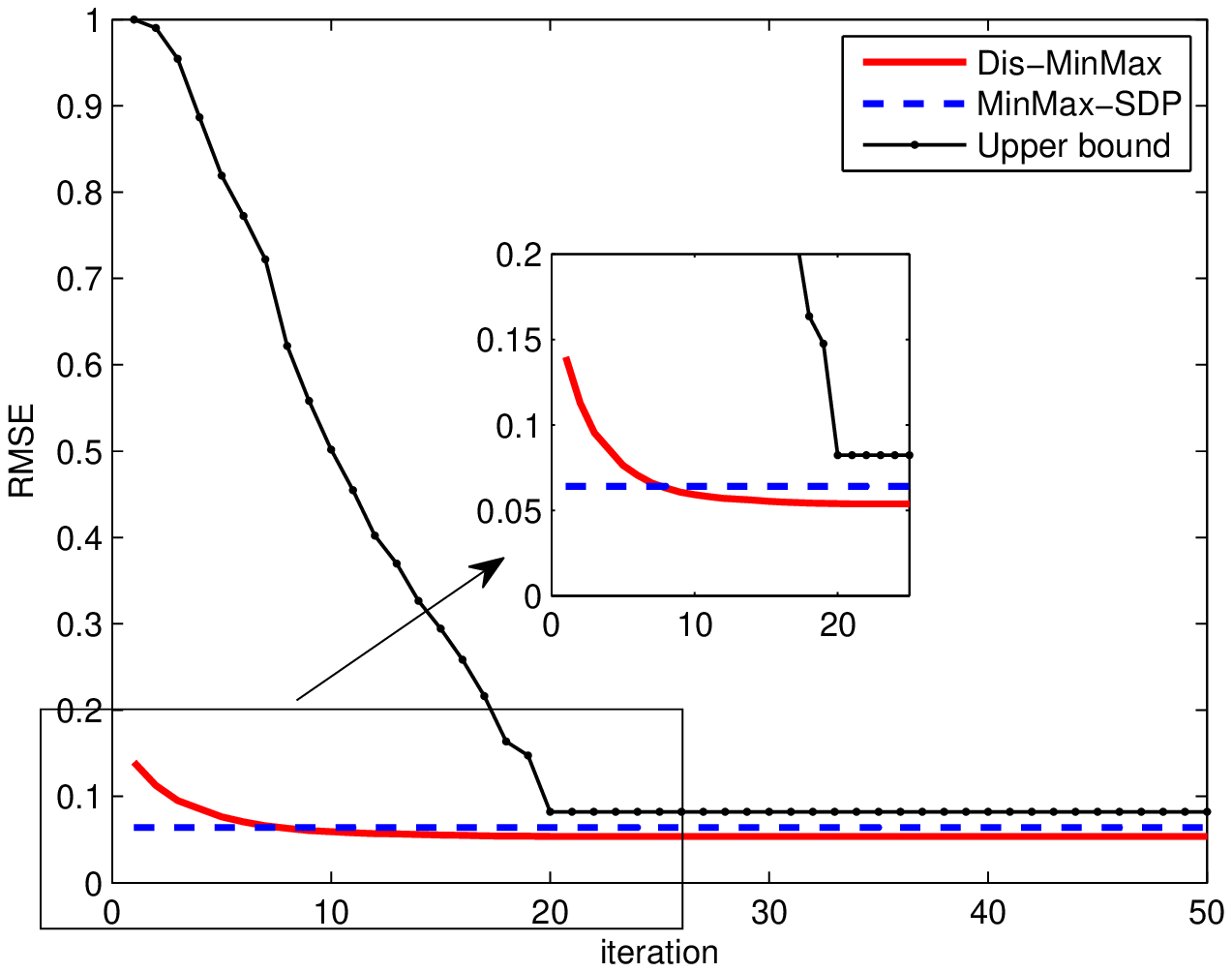}
  \label{subfig:distributedRMSE_N200}
 \end{minipage}}
 \caption{Localization performance of Dis-MinMax under Gaussian distributed measurement errors}
\label{fig:distributedPerformance}
\vspace{-20pt}
 \end{figure}

In the simulations about the performance of Dis-MinMax, we let four anchors be placed at $(-0.5,-0.5)$, $(0.5,-0.5)$, $(-0.5,0.5)$, $(0.5,0.5)$, and the sensors with unknown positions be randomly deployed using a uniform distribution in a unit square area. Three scenarios are considered: $a)$ $n=50$ and $R=0.5$; $b)$ $n=100$ and $R=0.3$; $c)$ $n=200$ and $R=0.2$. The measurement errors follow a zero-mean Gaussian distribution, with standard deviation as $\sigma=0.02$. The error bound is set as $\gamma=3\sigma$. Fig.~\ref{subfig:distributedRMSE_N50} compares the RMSEs of the position estimates obtained by MinMax-SDP and Dis-MinMax. We can see the RMSEs obtained by Dis-MinMax converge very quickly and the converged value is slightly smaller than that obtained by the centralized algorithm. Though MinMax-SDP and Dis-MinMax are both minimizing the worst-case estimation error, they are tackled in different ways, which results in the difference between the converged value obtained by the distributed algorithm and the RMSEs obtained by the centralized algorithm. Fig.~\ref{subfig:distributedRMSE_N50} also illustrates the iteration process of the upper bound for the RMSE, which is defined as $\sqrt{\frac{\sum_{i=1}^{n}R_i(\tau)^2}{n}}$. From Fig.~\ref{subfig:distributedRMSE_N50}, we can see that the upper bound for the RMSE also quickly converges. The localization results when the number of sensors increases to 100 and 200 are shown in Fig.\ref{subfig:distributedRMSE_N100} and Fig.\ref{subfig:distributedRMSE_N200}, from which, we can observe similar performance with that shown in Fig.~\ref{subfig:distributedRMSE_N50}. Therefore, Dis-MinMax is scalable to a larger number of nodes.

\begin{table}[!h]
\tabcolsep 1.2mm
\renewcommand{\arraystretch}{1.4}
\caption{Analytical comparison of different distributed algorithms} \label{table:analyticalCompare}
\vspace{-10pt}
\begin{center}
\begin{tabular}{|c|c|c|c|c|c|}
\hline
\tabincell{c}{} & \tabincell{c}{Dis-MinMax } & \tabincell{c}{E-ML with ADMM} & \tabincell{c}{ECHO}& \tabincell{c}{SNLRS}\\
\hline
       Size of convex problem &      $2|\mathcal{N}_{i}|+2$ & $7|\mathcal{N}_{ij}|+2|\mathcal{N}_{ia}|+3$    & not applicable&  not applicable\\
\hline
       Computational complexity &      $\mathcal{{O}}(|\mathcal{N}_{i}|^3)$ &      $\mathcal{{O}}(|\mathcal{N}_{i}|^3)$                                 &$\mathcal{{O}}(|\mathcal{N}_{i}|)$ &$\mathcal{{O}}(|\mathcal{N}_{i}|)$\\
\hline
       Communication cost &      $2|\mathcal{N}_{ij}|$ &      $9|\mathcal{N}_{ij}|$ &   $2|\mathcal{N}_{ij}|$&   $|\mathcal{N}_{ij}|$\\
\hline
\end{tabular}
\end{center}
\vspace{-20pt}
\end{table}

We then compare the localization performance of Dis-MinMax with other exiting distributed algorithms:  $a)$ E-ML with ADMM in \cite{MLE2014TSP}, which is a distributed implementation of aforementioned E-ML and can converge with a sublinear rate; $b)$ ECHO in \cite{diao2014barycentric}, which uses barycentric coordinates to express the positions in a linear form and can converge to the true positions in error-free case; $c)$ SNLRS in \cite{ICASSP2015}, which divides the entire network into several overlapping subnetworks, in which the sensors are localized via SDP, respectively. The global coordinates of sensor position estimates in each subnetwork are obtained through rigid registration. SNLRS is a variant of distributed SDP. Table \ref{table:analyticalCompare} compares the computational complexities and communication costs \textbf{per iteration} of different distributed algorithms. In this table, $\mid\mathcal{N}_i\mid$ denotes the number of sensor $i$'s neighboring nodes, $|\mathcal{N}_{ij}|$ denotes the number of sensors among sensor $i$'s neighboring nodes, and $|\mathcal{N}_{ia}|$ denotes the number of anchors among sensor $i$'s neighboring nodes. We can find the computational complexity of ECHO is smaller than that of Dis-MinMax and E-ML with ADMM. The computational complexities of Dis-MinMax and E-ML with ADMM are of the same order. The communication costs of Dis-MinMax and ECHO are equal and smaller than that of E-ML with ADMM. It should be noted that in SNLRS, the sensors in each subnetwork are actually locally localized in a centralized way. With regard to SNLRS, table \ref{table:analyticalCompare} only shows the computational complexities and communication costs \textbf{per iteration} during the position refinement using gradient-based search after the global rigid registration. Fig.~\ref{fig:RMSE_compare} compares the localization results of these distributed algorithms. In Fig.~\ref{fig:RMSE_compare}, both Dis-MinMax and E-ML with ADMM can converge very fast, SNLRS converges in more than $100$ iterations and ECHO converges in more than $10^5$ iterations. Moreover, the converged value of RMSE obtained by Dis-MinMax is much lower than those obtained by other three algorithms. We can conclude that the performance of Dis-MinMax is the best among these four algorithms.
  \begin{figure}[!ht]
  \centering
 \includegraphics[width=0.35\textwidth]{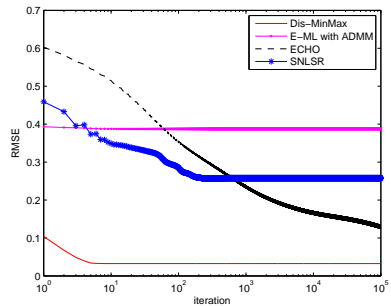}
  \vspace{-15pt}
 \caption{Comparison of RMSEs obtained by Dis-MinMax, E-ML with ADMM, ECHO and SNLRS under Gaussian distributed measurement errors, where $\sigma=0.02$, $n=100$ and $R=0.3$.}
\label{fig:RMSE_compare}
\vspace{-20pt}
 \end{figure}

 \section{Conclusions}\label{sec:conclusion}
In this paper, we investigate a network localization problem with unknown and bounded measurement errors. We formulate this problem as a non-convex optimization problem to minimize the worst-case localization error. Through relaxation, we transform the non-convex optimization problem into a convex optimization problem, whose dual problem can be solved through semidefinite programming. We give a geometrical interpretation of our problem and prove that the localization error of our proposed algorithm is upper bounded. Furthermore, we propose a distributed algorithm, along with an initial estimation algorithm. The convergence of the distributed algorithm is also proved. Extensive simulations show that both the centralized MinMax-SDP and Dis-MinMax can perform very well without the statistical knowledge of measurement errors.

In this paper, we only consider the localization problem with bounded range measurements. One interesting extension of this work would be studying the localization problem with other forms of measurements, e.g., AOA, RSS, TDOA, etc., or hybrid measurements with unknown and bounded measurement errors.


\end{document}